\pdfoutput=1
\documentclass[12pt,draftcls,onecolumn]{IEEEtran}

\IEEEoverridecommandlockouts                              

\ifCLASSINFOpdf
  \usepackage[pdftex]{graphicx}
  \graphicspath{{../pdf/}{../jpeg/}}
   \DeclareGraphicsExtensions{.pdf,.jpeg,.png}
\else
  \usepackage[dvips]{graphicx}
   \graphicspath{{../eps/}}
   \DeclareGraphicsExtensions{.eps}
\fi

\usepackage{epsfig} 

\usepackage{setspace}
\usepackage{amsmath,amssymb,array,graphicx,verbatim}

\usepackage{amsthm}
\usepackage{blkarray}
\usepackage{rotating}

\usepackage{enumerate}
\usepackage[shortlabels]{enumitem}
\newcommand{\mylabel}[2]{#1#2}

\newtheorem{theorem}{Theorem}
\newtheorem{lemma}{Lemma}
\newtheorem{problem}{Problem}
\newtheorem*{pb_mjls*}{Finite Horizon MJLS Optimal Control Problem}

\newtheorem{claim}{Claim}
\newtheorem{corollary}{Corollary}
\newtheorem{definition}{Definition}

\newtheorem{assumption}{Assumption}
\newtheorem{remark}{Remark}
\usepackage{enumerate}
\usepackage{epstopdf}
\usepackage[noadjust]{cite}
  
\usepackage{dblfloatfix}
\usepackage{dsfont}
\usepackage{mathrsfs}
\usepackage{relsize}
\usepackage{mathtools}
\usepackage{url}

\usepackage{tikz}
\usepackage{circuitikz}
\usepackage{tikz-cd}
\usetikzlibrary{arrows,shapes,calc,positioning}
\tikzstyle{block} = [draw,rectangle, rounded corners, minimum width=1cm, minimum height=0.8cm,text centered, line width=2pt ]
\tikzstyle{arrow} = [thick,->,>=stealth,line width=2pt]

\tikzset{cross/.style={cross out, draw=black, minimum size=2*(#1-\pgflinewidth), inner sep=0pt, outer sep=0pt},
cross/.default={1pt}}

\usepackage{enumerate}
\tikzset{
  shift left/.style ={commutative diagrams/shift left={#1}},
  shift right/.style={commutative diagrams/shift right={#1}}
}
\usetikzlibrary{calc}

\newcommand\rsmraise[1]{%
  \ifx#1\displaystyle .8\else
    \ifx#1\textstyle .8\else
      \ifx#1\scriptstyle .6\else
        .45%
      \fi
    \fi
  \fi}

\usepackage{epstopdf}
\usetikzlibrary{shapes}
\tikzstyle{block} = [draw,rectangle, rounded corners, minimum width=1cm, minimum height=0.8cm,text centered, line width=2pt ]
\tikzstyle{arrow} = [thick,->,>=stealth,line width=2pt]

\usetikzlibrary{arrows,decorations.markings,decorations}
\tikzset{
    addarrow/.style={decoration={markings, mark=at position 1 with {\arrow{stealth}}},
                     postaction={decorate}}
}

\newcommand{\tp}{\intercal}		
\newcommand{\R}{\mathbb{R}}			

\DeclareMathOperator{\ee}{\mathbb{E}}			
\DeclareMathOperator{\prob}{\mathbb{P}}			
\DeclareMathOperator{\vecc}{\mathbf{vec}}		
\DeclareMathOperator{\tr}{\mathbf{tr}}			
\DeclareMathOperator{\cov}{\mathbf{cov}}		
\DeclareMathOperator{\diag}{\mathbf{diag}}	

\newcommand{\bmat}[1]{\begin{bmatrix}#1\end{bmatrix}}

\newcounter{l1}
\newcounter{l2}
\newcommand{\bdashlist}{\begin{list}{$-$}{} }
\newcommand{\bromalist}{\begin{list}{\roman{l1}.}{\usecounter{l1}}}
\newcommand{\balphlist}{\begin{list}{(\alph{l2})}{\usecounter{l2}}}

\title{
Optimal Infinite Horizon Decentralized Networked Controllers with Unreliable Communication
}

\author{Yi Ouyang, Seyed Mohammad Asghari, and Ashutosh Nayyar
\thanks{Preliminary version of this paper appears in the proceedings of the 2018 American Control Conference (ACC) (see \cite{Ouyang_Asghari_Nayyar:2018}).}
\thanks{
Y. Ouyang is with the University of California, Berkeley. 
S. M. Asghari and A. Nayyar are with the Department of Electrical
Engineering, University of Southern California, Los Angeles, CA. 
Email: ouyangyi@berkeley.edu; asgharip@usc.edu; ashutosn@usc.edu.
}
}

\begin{document}

\maketitle

\begin{abstract}
We consider a decentralized networked control system (DNCS) consisting of a remote controller and a collection of linear plants, each
associated with a local controller. Each local controller directly observes the state of its co-located plant and can inform the remote controller of the plant's state through an unreliable uplink channel. The downlink channels from the remote controller to local controllers were assumed to be perfect. The objective of the local controllers and the remote controller is to cooperatively minimize the infinite horizon time average of expected quadratic cost. The finite horizon version of this problem was solved in our prior work \cite{asghari_ouyang_nayyar_tac_2018}. The optimal strategies in the finite horizon case were shown to be characterized by coupled Riccati recursions.  In this paper, we show that if the link failure probabilities are below certain critical thresholds, then the coupled Riccati recursions of the finite horizon solution reach a steady state and the corresponding decentralized strategies are optimal. Above these thresholds, we show that no strategy can achieve finite cost. We exploit a connection between our DNCS Riccati recursions and the coupled Riccati recursions of an auxiliary Markov jump linear system to obtain our results. Our main results  in Theorems \ref{thm:DC_infinite_2C} and \ref{thm:DC_infinite_NC} explicitly identify the critical thresholds for the link failure probabilities and  the  optimal decentralized control strategies  when all link failure probabilities are below their thresholds. 
\end{abstract}


\section{Introduction}
\label{sec:intro}

%

Many cyber-physical systems can be viewed as Networked Control Systems (NCSs) consisting of  several components such as physical systems, controllers, actuators and sensors that are interconnected by communication networks.
One key question in the design and operation of such systems is the following: what effect do communication limitations and imperfections such as packet loss, delays, noise and data rate limits have on the system performance? A well-studied communication model in the context of NCSs is that of an unreliable communication link that randomly loses packets. This means that the receiver in this unreliable link (e.g., a controller, an actuator etc.) receives information intermittently and has to perform its functions (selecting a control action, applying a control on the plant etc.) despite the interruptions in communication.

Networked control and  estimation problems in which there is only a single controller in the NCS and the unreliable links are from sensor(s) to the controller and/or from the controller to actuator(s) have been a focus of significant research (see, for example, \cite{sinopoli2004kalman,Imer2006optimal,Sinopoli2005,Sinopoli2006,Elia2004, Garone2008,Schenato2007, Knorn_Dey_2015, Gupta_Martins_2009,Gupta_Dana_2009, Gupta_Hassibi_2007,BansalBasar:1989a,TatikondaSahaiMitter:2004,NairFagnaniZampieriEvan:2007,molin2013optimality,rabi2014separated} ).
In many complex NCSs, however, there are multiple controllers which may need to communicate with each other to control the overall system. In such cases, the unreliable communication may not be just between sensors and controllers or controllers and actuators but also among controllers themselves.
Thus, multiple controllers may need to make decentralized decisions while communicating intermittently with each other and with the sensors and actuators of the system. We will refer to such a NCS as a Decentralized Networked Control System (DNCS) since the control decisions need to be made in a decentralized manner.

The fact that multiple controllers need to make decentralized decisions means that control problems in DNCSs  can be viewed as decentralized control problems. 
Optimal decentralized control problems are generally difficult to solve (see \cite{Witsenhausen:1968, LipsaMartins:2011b,blondel2000survey, mahajan_martins_yuksel}). 
In general, linear control strategies may not be optimal, and even the problem of finding the best linear control strategies may not be a  convex problem  \cite{linearPHS}, \cite{YukselBasar:2013}. 
Existing methods for computing optimal decentralized controllers require specific information structures and system properties such as  partial nestedness of the information structure \cite{HoChu:1972}, stochastic nestedness \cite{Yuksel:2009} quadratic invariance \cite{RotkowitzLall:2006}, substitutability \cite{AsghariNayyar:2015, asghari2016dynamic} etc. A common feature of the prior work in decentralized control is that the underlying communication structure of the decentralized system is assumed to be fixed and unchanging. For example, several works assume a fixed communication graph among controllers whose (directed) edges represent perfect communication links between controllers
\cite{SwigartLall:2010, ShahParrilo:2013, kim2012separable, swigart_lall_2011, LessardLall:2011, lessard_lall_2012, Lessard_Lall_2015, Lessard_2012, tanaka2014optimal, NayyarL14}. Similarly, when the communication graph incorporates delays, the delays are assumed to be fixed \cite{VaraiyaWalrand:1978, Kurtaran_Sivan_1974, Yoshikawa:1975, Nayyar_mahajan_Teneketzis_2011, NayyarKalathilJain:conf, Rantzer:2007, LamperskiDoyle:2011, FeyzmahdavianGattamiJohansson:2012, lamperski_doyle_2015,lamperski2015optimal}. Such models, however, do not incorporate the intermittent nature of communication over unreliable links between controllers. 
While some works \cite{chang2011explicit,chang2011synthesis} have investigated  unreliable controller-actuator communication  in the context of a decentralized control problem, they require that the inter-controller communication be perfect.

In this paper, we investigate a decentralized control problem with unreliable inter-controller communication. In particular, we consider a DNCS consisting of a remote controller and a collection of linear plants, each
associated with a local controller. Each plant is directly controlled by a local controller which can perfectly observe the state of the plant. The remote controller can control all plants, but it does not have direct access to the states as its name suggests. The remote controller and the local controllers are connected
by a communication network in which the downlinks from  remote controller to  local controllers are perfect but the uplinks
from local controllers to  remote controller are unreliable channels with random packet drops.  The objective of the local controllers and the remote controller is to cooperatively
minimize an overall quadratic performance cost of the DNCS. The information structure of this DNCS does not fit into the standard definition of partially nested information structures due to the unreliable links between controllers.  

For the finite horizon version of our  problem, 
we obtained optimal  decentralized controllers  in \cite{ouyang2016optimal, asghari_ouyang_nayyar_tac_2018} using ideas from the common information approach \cite{nayyar2013decentralized}. The optimal strategies in the finite horizon case were shown to be characterized by coupled Riccati recursions. Another approach based on  Pontryagin's maximum principle was used in \cite{liang2018control} for the finite horizon problem with only two controllers. In this paper, we will focus on the infinite time horizon average cost problem. The infinite horizon problem differs from its finite horizon counterpart in several key ways: \\
(i) In the finite horizon problem, the optimal cost is always finite. In the infinite horizon problem, however, it may be the case that no strategy can achieve a finite cost over the infinite horizon. In fact, we will show that this is the case  if the link failure probabilities are above certain thresholds. \\
(ii) Similarly, the finite horizon problem does not have to deal with the issue of stability since under any reasonable finite horizon strategy the system state cannot become ``too large'' in a finite time. The stability of the state becomes a key issue in the infinite horizon. In addition to proving optimality of control strategies, we need to make sure that the optimal strategies keep the state mean-square stable. \\
(iii) Finally, the analytical approaches for the finite and infinite horizon problems are fundamentally different. In the finite horizon case, we were able to use the common information approach to obtain a coordinator-based dynamic program. In the infinite horizon case, our essential task is to show that the value functions of the coordinator-based finite horizon dynamic program converge to a steady state as the horizon approaches infinity. Since the value functions were characterized by coupled Riccati recursions, this boils down to showing that these coupled recursions reach a steady state.  Further, we need to show that the decentralized control strategies characterized by the steady-state coupled Riccati equations are indeed optimal. We achieve these goals by establishing a connection between our DNCS and an auxiliary (and fictitious) Markov jump linear system (MJLS)\footnote{Note  that due to the presence of multiple controllers, our DNCS cannot be viewed as a standard MJLS (with one controller). Nevertheless, we show that it is still possible to use some MJLS results for our DNCS.}.
An alternative  approach for the two-controller version of our infinite horizon problem was used  in \cite{liang2018control} to find optimal strategies if certain coupled Riccati equations have solutions.

%

\subsection{Contributions of the Paper}
\begin{enumerate}
\item  We investigate an infinite time horizon decentralized stochastic control problem in which local controllers send their information to a
remote controller over unreliable links. To the best of our
knowledge, this is the first paper that solves an infinite time horizon optimal
decentralized control problem with unreliable communication
between controllers.  The finite time horizon version of our problem was solved in \cite{ouyang2016optimal, asghari_ouyang_nayyar_tac_2018} and our results  in this paper use the finite horizon solutions obtained there. However, unlike the finite horizon case, we have to address the possibility that no control strategy may achieve finite cost over infinite time horizon. Due to such stability related issues, our approach for the infinite horizon problem is markedly different from the  common information based approach adopted in \cite{ouyang2016optimal, asghari_ouyang_nayyar_tac_2018}.

\item We show that there are critical thresholds for link failure probabilities above which no control strategy can achieve a finite cost in our problem. When the link failure probabilities are  below their critical thresholds, we show that the optimal control strategies of this infinite horizon decentralized control problem admit simple structures: the optimal remote controller strategy is a time-invariant linear function of the common estimates of system states and the optimal strategies for  local controllers are time-invariant linear functions of the common estimates of system states and the perfectly observed local states. The main strengths of our result are that (i) it provides simple strategies that are proven to be optimal: not only are the strategies in Theorems \ref{thm:DC_infinite_2C} and \ref{thm:DC_infinite_NC} linear, they use estimates that can be easily updated; (ii) it shows that the optimal strategies are completely characterized by solution of coupled Riccati equations. 


\item If the local controllers act only as sensors and the remote controller is the only controller in the system, then our model reduces to a NCS with multiple sensors  observing different components of the system state and communicating with the remote controller over independent unreliable links. Thus, we obtain optimal strategy and critical probabilities for a multi-sensor, single-controller NCS as a corollary of our result in Theorem \ref{thm:DC_infinite_NC}.
 
\item Finally, our problem can be viewed as a dynamic team problem  by viewing  each controller's actions at different time instants as the actions of distinct players \cite{HoChu:1972}. Since we are interested in infinite time horizon, this team-theoretic viewpoint means that our dynamic team has infinitely many players. Further, due to the unreliable links, our problem does not directly fit into the partially nested LQG team problem. Thus, the standard results for partially nested LQG teams with finitely many players \cite{HoChu:1972} do no apply to our problem.  As observed in \cite{infinite_team}, results for teams with finitely many players cannot be directly extended to teams with infinitely many players even when the information structure is static or partially nested.
 In spite of this,  our dynamic team problem turns out to have simple optimal strategies.
\end{enumerate}


\subsection{Organization}
The rest of the paper is organized as follows. Section \ref{sec:pre} summarizes the notations and operators used in this paper. 
In Section \ref{sec:model_2_controllers}, we formulate the finite horizon and infinite horizon optimal control problems for a DNCS with one remote controller and one local controller. We briefly review Markov Jump Linear Systems (MJLSs) in Section \ref{sec:mjls}. We establish a connection between the  DNCS of Section \ref{sec:model_2_controllers} and an auxiliary MJLS in Section \ref{sec:infinite_2_controllers} and use this connection to provide our main results for the  DNCS of Section \ref{sec:model_2_controllers}. In Section \ref{sec:model_N_controllers}, we extend our DNCS model to the case with multiple local controllers and provide our main results for  this DNCS. We discuss some key aspects of our approach in Section \ref{sec:discussion}. Section \ref{sec:conclusion} concludes the paper.
The proofs of all  technical results are in the Appendices.


\section{Preliminaries}
\label{sec:pre}
\subsection{Notations}
In general, subscripts are used as time indices while superscripts are used to index controllers.
For time indices $t_1\leq t_2$, $X_{t_1:t_2}$ is a short hand notation for the collection variables $(X_{t_1},X_{t_1+1},...,X_{t_2})$.
Random variables/vectors are denoted by upper case letters, their realizations by the corresponding lower case letters.
For a sequence of column vectors $X, Y, Z, \ldots$, the notation $\vecc(X,Y,Z,\ldots)$ denotes the vector $[X^{\tp}, Y^{\tp}, Z^{\tp},...]^{\tp}$. $\prob(\cdot)$ denotes the probability of an event, and $\ee[\cdot]$ and $\cov(\cdot)$ denote the expectation and the covariance matrix of a random variable/vector. The transpose, trace, and spectral radius of a matrix $A$ are denoted by $A^{\tp}$, $\tr(A)$, and $\rho(A)$, respectively. 
For two symmetric matrices $A,B$, $A\succeq B$ (resp. $A\succ B$) means that $(A-B)$ is positive semi-definite (PSD) (resp. positive definite (PD)).
For a block matrix $A$, we use $[A]_{m,:}$ to denote  the $m$-th block row and $[A]_{:,n}$ to denote  the $n$-th block column of $A$. Further, $[A]_{m,n}$ denotes the block  located at the $m$-th block row and $n$-th block column of $A$. For example, if 
\begin{align}
A = \begin{bmatrix}
A^{11} & A^{12} & A^{13} \\
A^{21} & A^{22} & A^{23} \\
A^{31} & A^{32} & A^{33} 
\end{bmatrix},\notag
\end{align}
then $[A]_{2,:} = \begin{bmatrix}
A^{21} & A^{22} & A^{23}
\end{bmatrix}$, $[A]_{:,3} = \begin{bmatrix}
A^{13} \\ A^{23} \\ A^{33}
\end{bmatrix}$, and $[A]_{2,3} = A^{23}$.
We use $\R^n$ to denote the $n$-dimensional Euclidean space and $\R^{m \times n}$ to denote the space of all real-valued $m \times n$ matrices. We use $\otimes$ to denote the Kronecker product.

\subsection{Operator Definitions}\label{sec:operators}
We define the following operators. 
\begin{itemize}
\item Consider matrices $P,Q,R,A,B$ of appropriate dimensions with $P,Q$ being PSD matrices and $R$ being  a PD matrix. We define $\Omega(P,Q,R,A,B)$ and $\Psi(P,R,A,B)$ as follows: 
\begin{align}
\label{Omega}
&\Omega(P,Q,R,A,B) 
:= Q+A^\tp P A- \notag\\
& ~~~~~~~~ ~~~~~~~~ ~~~~~~~~A^\tp P B(R+B^\tp P B)^{-1}B^\tp P A.
\\
&\Psi(P,R,A,B) 
:=
-(R+B^\tp P B)^{-1}B^\tp P A.
\label{Psi}
\end{align}
Note that $P = \Omega(P,Q,R,A,B) $ is the discrete time algebraic Riccati equation.

\item Let $P$ be a block matrix with $M_1$ block rows and $M_2$ block columns. Then, for numbers $m_1, m_2$ and matrix $Q$, $\mathcal{L}_{zero}(P,Q,m_1,m_2)$ is a matrix with the same size as $P$ defined as follows:
\begin{align}
&\mathcal{L}_{zero}(P,Q,m_1,m_2) := \notag \\
&\begin{blockarray}{cccl}
\text{$1:m_2-1$} &m_2 &\text{$m_2+1:M_2$}  &  \\
\begin{block}{[ccc]l}
                       \mathbf{0} & \mathbf{0} & \mathbf{0} & \text{$1:m_1-1$} \\
   \mathbf{0} & Q & \mathbf{0} &m_1 \\
                    \mathbf{0} & \mathbf{0} & \mathbf{0} &\text{$m_1+1:M_1$}  \\
\end{block}
\end{blockarray}
\label{L_zero}
\end{align}

\item Let $P$ be a block matrix with $M_1$ block rows and $M_1$ block columns. Then, for number $m_1$ and matrix $Q$, $\mathcal{L}_{iden}(P,Q,m_1)$ is a matrix with the same size as $P$ defined as follows:
\begin{align}
&\mathcal{L}_{iden}(P,Q,m_1) := \notag \\
&\begin{blockarray}{cccl}
\text{$1:m_1-1$} &m_1 &\text{$m_1+1:M_1$}  &  \\
\begin{block}{[ccc]l}
                       \mathbf{I} & \mathbf{0} & \mathbf{0} & \text{$1:m_1-1$} \\
   \mathbf{0} & Q & \mathbf{0} &m_1 \\
                    \mathbf{0} & \mathbf{0} & \mathbf{I} &\text{$m_1+1:M_1$}  \\
\end{block}
\end{blockarray}
\label{L_iden}
\end{align}
\end{itemize}

\section{System Model and Problem Formulation}
\label{sec:model_2_controllers}
\begin{figure}
\begin{center}
\begin{tikzpicture}
\node [rectangle,draw,minimum width=1.5cm,minimum height=1cm,line width=1pt,rounded corners]at (-1.1,-0.1) (1) {
\begin{small}
\begin{tabular}{c}
Remote \\ Controller $C^0$
\end{tabular}
\end{small}}; 
\node [rectangle,draw,minimum width=1.5cm,minimum height=1cm,line width=1pt,rounded corners]at (5.1,-0.1) (2) {
\begin{small}
\begin{tabular}{c}
Local \\ Controller $C^1$
\end{tabular}
\end{small}
}; 
\node [rectangle,draw,minimum width=7cm,minimum height=1cm,line width=1pt,rounded corners]at (2,2) (3) {Plant}; 


\draw [line width=1pt] (3.65,-0.15) to[cspst] (0.5,-0.15);
\draw [->,line width=1pt,>=stealth]  (0.5,-0.15) -- (0.30,-0.15);

\draw [->,line width=1pt]  (0.30,0.15) -- (3.65,0.15);

\draw [->,line width=1pt,dashed,>=stealth] (5,0.5) --  (5,1.5);   
\draw [->,line width=1pt,>=stealth] (4.7,1.5)  -- (4.7,0.5); 
\draw [->,line width=1pt,dashed,>=stealth] (-1,0.5) -- (-1,1.5);  

\node[] at (4.4,1) {$X_t$};
\node[] at (5.4,1) {$U_t^1$};
\node[] at (-0.6,1) {$U_t^0$};
\node[] at (2.3,-0.6) {$\Gamma_t^1$};
\node[] at (3.2,-0.4) {$X_t$};
\node[] at (0.9,-0.4) {$Z_t^1$};
\node[] at (2.1,0.4) {$Z_t^1, U_{t-1}^0$};

\end{tikzpicture}
\caption{Two-controller system model. The binary random variable $\Gamma_t^1$ indicates whether packets are transmitted successfully. Dashed lines indicate control links and solid lines indicate communication links.}
\label{fig:SystemModel_2C}
\end{center}
\end{figure}
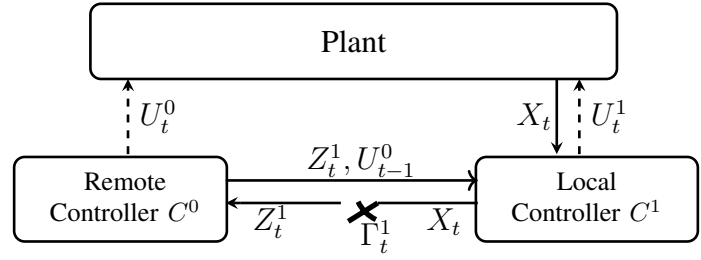

Consider a discrete-time system with a local controller $C^1$ and a remote controller $C^0$  as shown in Fig. \ref{fig:SystemModel_2C}. 
The linear plant dynamics are given by
\begin{align}
X_{t+1} = & A X_t +  B^{10} U^0_t+ B^{11}U^1_t + W_t
\notag\\
= & A X_t +  B U_t + W_t,
 \label{Model:system_2C}
\end{align}
where $X_t\in \R^{d_X}$ is the state of the plant at time $t$, $U_t = \vecc(U^0_t,U^1_t)$, 
$U^0_t\in \R^{d^0_U}$ is the control action of the remote controller $C^0$, $U^1_t \in \R^{d^1_U}$ is the control action of the local controller $C^1$, and $W_t$ is the noise at time $t$.  $A, B := [B^{10},B^{11}]$ are matrices with appropriate dimensions.We assume that $X_0 = 0$ and $W_t, t=0,1,\dots,$ is an i.i.d. noise process with $\cov(W_t) = \mathbf I$.


\subsection{Communication Model}
\label{subs:comm_model_2C}
 At each time $t$, the local controller $C^1$  observes the state $X_t$ perfectly and sends the observed state to the remote controller $C^0$ through an unreliable link with packet drop probability $p^1$. 
 Let $\Gamma_t^1$ be a Bernoulli random variable describing the state of this link, that is, $\Gamma_t^1=0$ if the link is broken (i.e., the packet is dropped) and  $\Gamma_t^1=1$ if the link is active. We assume that $\Gamma_{t}^1, t \geq 0,$ is an i.i.d. process and is independent of the noise process $W_{0:t}, t \geq 0$.  Let $Z_t^1$ be the output of the unreliable link. Then,
 \begin{align}
\Gamma_t^1 = &\left\{\begin{array}{ll}
1 & \text{ with probability }(1-p^1),\\
0 & \text{ with probability }p^1.
\end{array}\right. \label{eq:gamma}
\\
Z_t^1 = &
\left\{\begin{array}{ll}
X_t & \text{ when } \Gamma_t^1 = 1,\\
\emptyset & \text{ when } \Gamma_t^1 = 0.
\end{array}\right.
\label{Model:channel_2C}
\end{align}
We assume that $Z_t^1$ is perfectly observed by $C^0$. Further,  we assume that  $C^0$ sends an acknowledgment to the local controller $C^1$ if it receives the state value. Thus, effectively,  $Z_t^1$ is perfectly observed by $C^1$ as well.  The two controllers select their control actions at time $t$ after observing $Z_t^1$. We assume that the links for sending acknowledgments as well as the links from the controllers to the plant are perfectly reliable.

\subsection{Information structure and cost}
\label{subs:info_cost_2C}

Let $H_t^0$ and $H_t^1$ denote the information available to the controllers $C^0$ and $C^1$ to make decisions at time $t$, respectively. Then,
\begin{align}
H^0_t = \{Z_{0:t}^1, U^0_{0:t-1}\}, \hspace{2mm} H^1_t= \{X_{0:t}, Z_{0:t}^1, U^1_{0:t-1}, U^0_{0:t-1}\}. 
\label{Model:info_2C}
\end{align}
$H^0_t$ will be referred to as the \textit{common information} among the two controllers at time $t$\footnote{ We assumed that $U^0_{0:t-1}$  is pat of $H^1_t$. This is not a restriction because even if $U^0_{0:t-1}$ is not directly observed by $C^1$ at time $t$,  $C^1$ can still compute it using $C^0$'s strategy since it knows everything $C^0$ knows.}.

Let $\mathcal{H}^0_t$ and $\mathcal{H}^1_t$ be the spaces of all possible 
realizations of $H^0_t$ and $H^1_t$, respectively.
Then, the control actions are selected according to
\begin{align}
U^0_t = g^0_t(H^0_t), \hspace{2mm} U^1_t = g^1_t(H^1_t), 
\label{Model:strategy_2C}
\end{align}
where the control laws $g^0_t:\mathcal{H}^0_t \to \R^{d_0}$ and $g^1_t:\mathcal{H}^1_t \to \R^{d_1}$ are measurable mappings.
We use $g:=(g^0_{0},g^0_1,\dots,g^1_{0},g^1_1,\dots)$ to denote the control strategies of $C^0$ and $C^1$.

The instantaneous cost $c(X_t,U_t)$ of the system is a  quadratic function given by
\begin{align}
&c(X_t,U_t) = 
X_t^\tp Q X_t + U_t^\tp R U_t,
\label{Model:cost_2C}
\end{align}
where $Q$ is a symmetric positive semi-definite (PSD) matrix, and $R = \begin{bmatrix}
R^{00} & R^{01}\\R^{10} & R^{11}
\end{bmatrix}$ is a symmetric positive definite (PD) matrix.

\subsection{Problem Formulation}
\label{subs:prob_formulaiton_2C}

Let $\mathcal{G}$ denote the set of all possible control strategies of $C^0$ and $C^1$. The performance of control strategies $g$ over a
finite horizon $T$ is measured by the total expected cost\footnote{Because the cost function $c(X_t,U_t)$ is always non-negative, the expectation is well-defined on the the extended real line $\R \cup \{\infty\}$.}:
\begin{align}
J_T(g):=\ee^{g}\left[\sum_{t=0}^T c(X_t,U_t)\right].
\label{Model:J_T_2C}
\end{align}

We refer to the system described by \eqref{Model:system_2C}-\eqref{Model:cost_2C} as the \emph{decentralized networked control system} (DNCS).
We consider the problem of strategy optimization for the DNCS over finite and infinite time horizons. These two problems are formally defined below.
\begin{problem}[Finite Horizon DNCS Optimal Control]
\label{problem_finite_2C}
For the DNCS described by \eqref{Model:system_2C}-\eqref{Model:cost_2C}, determine decentralized control strategies $g$ that optimize the  total expected cost over a finite horizon of duration $T$. In other words, solve the following strategy optimization problem:
\begin{align}
\inf_{g \in\mathcal{G}} 
J_T(g).
\label{Model:obj_finite_2C}
\end{align}
\end{problem}

\begin{problem}[Infinite Horizon DNCS Optimal Control]
\label{problem_infinite_2C}
For the DNCS described by \eqref{Model:system_2C}-\eqref{Model:cost_2C}, find decentralized strategies $g$ that minimize the  infinite horizon average cost.  In other words, solve the following strategy optimization problem:
\begin{align}
&\inf_{g \in\mathcal{G}} 
J_{\infty}(g)
:= \inf_{g \in\mathcal{G}} 
\limsup_{T\rightarrow\infty} \frac{1}{T+1} J_{T}(g).
\label{Model:obj_infinite_2C}
\end{align}

\end{problem}

We make the following  standard assumption on the system and cost matrices \cite{Bertsekas:1995}. 
\begin{assumption}
\label{assum:det_stb_2C}
$(A,Q^{1/2})$ is detectable and $(A,B)$ is stabilizable.
\end{assumption}

The finite horizon DNCS optimal control problem (Problem \ref{problem_finite_2C}) has been solved in \cite{ouyang2016optimal, asghari_ouyang_nayyar_tac_2018}.  We summarize the finite horizon results below. 

\begin{lemma}(\cite[Theorem 2]{ouyang2016optimal})
\label{lm:opt_strategies_2C}
The optimal control strategies of Problem \ref{problem_finite_2C} are given by
\begin{align}
\bmat{U^{0*}_t \\ U^{1*}_t \\} = K_t^0\hat X_t + \bmat{\mathbf{0} \\ K_t^1 }\left(X_t - \hat X_t\right), \label{eq:opt_U_2C}
\end{align}
where $\hat X_t = \ee[X_t|H^0_t]$ is the  estimate (conditional expectation) of $X_t$ based on the common information $H^0_t$.  The  estimate can be computed recursively according to
\begin{align}
\hat X_0 = & 0,
\label{eq:estimator_0_2C}
\\
\hat X_{t+1}
= &\left\{
\begin{array}{ll}
 \big(A + B K^0_t\big)\hat X_t& \text{ if }Z_{t+1}= \emptyset,\\
 X_{t+1} & \text{ if }Z_{t+1} = X_{t+1}.
\end{array}\right.
\label{eq:estimator_t_2C}
\end{align}
The gain matrices are given by
\begin{align}
& K_t^0 = \Psi(P_{t+1}^0,R,A,B),
\label{eq:K_finite_2C}
\\
& K_t^1 = \Psi((1-p^1)P_{t+1}^0+p^1 P_{t+1}^1,R^{11},A,B^{11}),
\label{eq:tildeK_finite_2C}
\end{align}
where $P_t^0$ and $P_t^1$ are PSD matrices obtained recursively as follows:
\begin{align}
&P_{T+1}^0 =  P_{T+1}^1 = \mathbf{0},  \label{eq:P_initial}\\
&P_t^0 =  \Omega(P_{t+1}^0,Q,R,A,B),
\label{eq:P_finite_2C}
\\
&P_t^1 =  \Omega((1-p^1)P_{t+1}^0+p^1 P_{t+1}^1,Q,R^{11},A,B^{11}).
\label{eq:tildeP_finite_2C}
\end{align}
Furthermore, the optimal cost is given by
\begin{align}
J^*_{T} = & \sum_{t = 0}^T \tr \big((1-p^1)P_{t+1}^0+p^1 P_{t+1}^1 \big).
\label{eq:opcost_finite_2C}
\end{align}
\end{lemma}

\begin{remark}
Note that remote controller's action $U^{0*}_t$ in \eqref{eq:opt_U_2C} is a function  of $\hat{X}_t$ only while the local controller's action $U^{1*}_t$ is a function of both $\hat{X}_t$ and $X_t$. Further, as per \eqref{eq:estimator_0_2C} and \eqref{eq:estimator_t_2C}, $\hat{X}_t$ is computed recursively based only on the knowledge of $Z_{0:t}$.
\end{remark}

In this paper, we will focus on solving the infinite horizon problem (Problem \ref{problem_infinite_2C}). Our solution will employ results from Markov Jump Linear Systems (MJLS). We provide a review of the relevant results from the theory of Markov jump linear systems before describing our solution to Problem \ref{problem_infinite_2C}.

\section{Review of Markov Jump Linear Systems}
\label{sec:mjls}
A discrete-time Markov Jump Linear System (MJLS) is described by the dynamics
\begin{align}
&X_{t+1}^{\diamond} = A^{\diamond}(M_t) X_t^{\diamond} + B^{\diamond}(M_t) U_t^{\diamond},
\label{eq:MJ_X}
\end{align}
where $X^{\diamond}_t  \in \R^{d_X}$ represents the state, $M_t \in \mathcal{M} = \{0,1,\ldots,M\}$ the mode, and $U^{\diamond}_t$ the control action at time $t$. $A^{\diamond}(M_t), B^{\diamond}(M_t)$ are mode-dependent matrices.
The mode $M_t$ evolves as a Markov chain described by the transition probability matrix $\Theta = [\theta^{ij}]_{i,j \in \mathcal{M}}$  such that
\begin{align}
\prob(M_{t+1} = j | M_t = i) = \theta^{ij}.
\label{eq:MJ_M}
\end{align}
The initial state $X_0^{\diamond}$ and mode $M_0$ are independent and they have  probability distributions $\pi_{X_0^{\diamond}}$ and $\pi_{M_0}$, respectively.

The information available at time $t$ to the controller of MJLS is  $H^{\diamond}_t = \{X^{\diamond}_{0:t},M_{0:t}, U^{\diamond}_{0:t-1}\}$.  The instantaneous cost incurred at time $t$ is given by
\begin{align}
&c^{\diamond}(X_t^{\diamond},U_t^{\diamond},M_t)  =
(X_t^{\diamond})^\tp Q^{\diamond}(M_t) X_t^{\diamond} + (U_t^{\diamond})^\tp R^{\diamond}(M_t) U_t^{\diamond},
\label{eq:MJ_c}
\end{align}
where $Q^{\diamond}(M_t), R^{\diamond}(M_t)$ are mode-dependent matrices.
An admissible control strategy is a sequence of measurable mappings $g^{\diamond} = (g^{\diamond}_0,g^{\diamond}_1,\dots)$ such that
$U^{\diamond}_t = g^{\diamond}_t(H^{\diamond}_t)$
and each $U^{\diamond}_t$ has finite second moment. Let $\mathcal G^{\diamond}$ be the set of all admissible control strategies for the MJLS.

The MJLS has been extensively studied in the literature (see \cite{costa2006discrete} and references therein). In the following, we state the finite horizon  optimal control problem for the MJLS and provide the optimal control strategy for this problem. 

\begin{pb_mjls*}
\label{problem_MJ_finite}
For the MJLS described by \eqref{eq:MJ_X}-\eqref{eq:MJ_c}, solve the following finite horizon strategy optimization problem
\begin{align}
\inf_{g^{\diamond} \in \mathcal G^{\diamond}}\sum_{t=0}^T \ee^{g^{\diamond}}\Big[ c^{\diamond}(X_t^{\diamond},U_t^{\diamond},M_t) \Big].
\end{align}

\end{pb_mjls*}

The solution of the finite horizon MJLS optimal control problem is given in the following lemma.
\begin{lemma}(\cite[Theorem 4.2]{costa2006discrete})
\label{opt_mjls_finite}
The optimal controller for finite horizon MJLS optimal control problem is given by
\begin{align}
U_t^{\diamond}  = K^{\diamond}_t(M_t) X_t^{\diamond},
\label{eq:optcontrol_MJ_finite}
\end{align}
where $K^{\diamond}_t(M_t)$ is a mode-dependent gain matrix. 
The gain matrices $K^{\diamond}_t(m)$, $m \in \mathcal{M}$, are given by
\begin{align}
& K^{\diamond}_t(m) = \Psi \big( \sum_{k=0}^M \theta^{mk} P^{\diamond}_{t+1}(k), R^{\diamond}(m),A^{\diamond}(m),B^{\diamond}(m)\big),
\label{eq:K_MJ_finite}
\end{align}
where the matrices  $P^{\diamond}_t(m)$,  $m \in \mathcal{M}$, are recursively computed as follows
\begin{align}
&P^{\diamond}_{T+1}(m) = \mathbf{0}, \label{eq:CARE_init} \\
&P^{\diamond}_t(m) =  \notag \\
&\Omega\big( \sum_{k=0}^M \theta^{mk} P^{\diamond}_{t+1}(k), Q^{\diamond}(m), R^{\diamond}(m),A^{\diamond}(m),B^{\diamond}(m)\big).
\label{eq:CARE_finite}
\end{align}
Further, the optimal cost is
\begin{align}
\ee [(X^{\diamond}_0)^\tp P^{\diamond}_{0}(M_0) X^{\diamond}_0].
\label{eq:optcost_MJ_finite}
\end{align}
\end{lemma}

One interesting property of the recursions in \eqref{eq:CARE_finite} is that under some certain stability conditions, matrices $P^{\diamond}_{t}(m)$, $m \in \mathcal{M}$, converge as $t \to -\infty$ to steady-state solutions $P^{\diamond}_{*}(m)$ satisfying \textit{discrete-time coupled algebraic Riccati equations (DCARE)}
\begin{align}
&P^{\diamond}_*(m) =  \notag \\
&\Omega\big( \sum_{k=0}^M \theta^{mk} P^{\diamond}_{*}(k), Q^{\diamond}(m), R^{\diamond}(m),A^{\diamond}(m),B^{\diamond}(m)\big).
\label{eq:CARE_infinite}
\end{align}
Before providing the results on the convergence of sequences of matrices $\{P^{\diamond}_{t}(m), t = T+1,T, T-1,\ldots \}$, $m \in \mathcal{M}$,   we introduce the concepts of stochastic stabilizability and stochastic detectability of the MJLS \cite{costa2006discrete}.
\begin{definition}
\label{def:ss}
The MJLS of \eqref{eq:MJ_X}-\eqref{eq:MJ_M} is stochastically stabilizable if there exist gain matrices $K^{\diamond}(m), m\in\mathcal M,$ such that for any initial state and mode, $\sum_{t=0}^\infty \ee[||X^{\diamond}_t||^2]< \infty$ where $X^{\diamond}_{t+1} = A_s(M_t) X^{\diamond}_t$ and  
\begin{align}
A_s (M_t) = A^{\diamond} (M_t) +  B^{\diamond}(M_t) K^{\diamond}(M_t).
\label{A_s_matrix}
\end{align}
In this case, we say the gain matrices $K^{\diamond}(m), m \in \mathcal M$, stabilize the MJLS.
\end{definition}

\begin{definition}
\label{def:sd}
The MJLS of \eqref{eq:MJ_X}-\eqref{eq:MJ_c} is Stochastically Detectable (SD) if there exist gain matrices $H^{\diamond}(m), m \in \mathcal{M}$, such that for any initial state and mode, $\sum_{t=0}^\infty \ee[||X^{\diamond}_t||^2]< \infty$ where 
$X^{\diamond}_{t+1} = A_d(M_t) X^{\diamond}_t$ and 
\begin{align}
A_d(M_t) = A^{\diamond} (M_t) +  H^{\diamond}(M_t) \big(Q^{\diamond}(M_t) \big)^{1/2}.
\label{A_d_matrix}
\end{align}
\end{definition}

From the theory of MJLS (\cite{costa2006discrete,costa1995discrete}), we can obtain the following result for 
the convergence of matrices $\{P^{\diamond}_{t}(m), t = T+1,T,T-1,\ldots \}$ to $P^{\diamond}_{*}(m)$ satisfying the DCARE in \eqref{eq:CARE_infinite}.

\begin{lemma}
\label{lm:MJ_infinite}
Suppose the MJLS is stochastically detectable (SD).  Then, matrices $P^{\diamond}_{t}(m)$, $m \in \mathcal{M}$, converge as $t \to -\infty$ to PSD matrices $P^{\diamond}_{*}(m)$ that satisfy the DCARE in \eqref{eq:CARE_infinite} if and only if the MJLS is stochastically stabilizable (SS).
%
%
\end{lemma}
\begin{proof}
See Appendix \ref{app:lm_MJ_infinite}.
\end{proof}

%

Stochastically stabilizability (SS) and stochastically detectability (SD) of a MJLS can be verified from the system matrices and the transition matrix for the mode of the MJLS \cite{costa2006discrete,fang2002stochastic}. Specifically, we have the following lemmas.

\begin{lemma}(\cite[Theorem 3.9]{costa2006discrete} and also \cite[Corollary 2.6]{fang2002stochastic})
\label{lm:ss}
\begin{enumerate}
\item A MJLS is SS if and only if there exist matrices $K^{\diamond}(m)$, $m \in \mathcal M$, such the matrix
\begin{align}
&\mathcal{A}_s := \notag \\
&\diag \big(A_s(0)\otimes A_s(0), \ldots, A_s(M)\otimes A_s(M) \big)
(\Theta^{\tp}\otimes \mathbf{I}),
\label{eq:bigmatrix_ss}
\end{align}
is Schur stable, i.e. $\rho(\mathcal{A}_s) < 1$, where $A_s(M_t)$ is given by \eqref{A_s_matrix}.
\item A MJLS is SD if and only if there exist matrices $H^{\diamond}(m)$, $m \in \mathcal M$, such the matrix
\begin{align}
&\mathcal{A}_d := \notag \\
&\diag \big(A_d(0)\otimes A_d(0), \ldots, A_d(M)\otimes A_d(M) \big)
(\Theta^{\tp}\otimes \mathbf{I}),
\label{eq:bigmatrix_sd}
\end{align}
is Schur stable, i.e. $\rho(\mathcal{A}_d) < 1$, where $A_d(M_t)$ is given by \eqref{A_d_matrix}.
\end{enumerate}
\end{lemma}

\section{Infinite Horizon Optimal Control}
\label{sec:infinite_2_controllers}
In centralized LQG control, the solution of the finite horizon problem can be used to solve the infinite horizon average cost problem by ensuring that the finite horizon Riccati recursions reach a steady state  and that the corresponding steady state strategies are optimal \cite{Bertsekas:1995}. We will follow a similar conceptual approach for our problem. Unlike the centralized LQG case, however, we have to ensure that coupled Riccati recursion reach a steady state. Even if such a steady state is reached, we need to show that the corresponding decentralized strategies outperform every other choice of decentralized strategies. Because of these issues, our analysis for the infinite horizon problem, Problem \ref{problem_infinite_2C}, will differ significantly from that of  centralized LQG problem. 

Recall that Lemma \ref{lm:opt_strategies_2C} in Section \ref{sec:model_2_controllers} describes the optimal control strategies for the finite horizon version of Problem \ref{problem_infinite_2C} (that is, Problem \ref{problem_finite_2C}).  Since we know the optimal control strategies for Problem \ref{problem_finite_2C}, solving Problem \ref{problem_infinite_2C} amounts to answering the following three questions: 
\begin{enumerate}[label=(\mylabel{Q}{\arabic*})]
\item Do matrices $P_t^0$ and $P_t^1$ defined in \eqref{eq:P_initial}-\eqref{eq:tildeP_finite_2C} converge as $t \to -\infty$ to $P_*^0$ and $P_*^1$ that satisfy the coupled fixed point  equations \eqref{eq:P_finite_2C_fixed} - \eqref{eq:tildeP_finite_2C_fixed} below? 
\begin{align}
&P_*^0 =  \Omega \big(P_{*}^0,Q,R,A,B\big),
\label{eq:P_finite_2C_fixed}
\\
&P_*^1 =  \Omega \big((1-p^1)P_{*}^0+p^1 P_{*}^1,Q,R^{11},A,B^{11}\big).
\label{eq:tildeP_finite_2C_fixed}
\end{align}
The above equations are steady state versions of \eqref{eq:P_finite_2C} - \eqref{eq:tildeP_finite_2C} obtained by replacing $P_t^0$, $P_{t+1}^0$ (resp. $P_t^1$, $P_{t+1}^1$ ) with  $P^0_*$ (resp. $P^1_*$).
\item If matrices $P_t^0$ and $P_t^1$ converge and we define matrices $K_*^0$ and $K_*^1$ using  $P_*^0$ and $P_*^1$ as follows, 
\begin{align}
& K_*^0 = \Psi(P_{*}^0,R,A,B),
\label{eq:K_finite_2C_fixed}
\\
& K_*^1 = \Psi \big((1-p^1)P_{*}^0+p^1 P_{*}^1,R^{11},A,B^{11}\big),
\label{eq:tildeK_finite_2C_fixed}
\end{align}
are the following strategies optimal for Problem \ref{problem_infinite_2C}?
\begin{align}
\bmat{U^{0*}_t \\ U^{1*}_t \\} = K_*^0\hat X_t + \bmat{\mathbf{0} \\ K_*^1 }\left(X_t - \hat X_t\right),
\label{eq:opt_U_2C_fixed}
\end{align}
where $\hat X_0 =  0$ and 
 \begin{align}
\hat X_{t+1}
= \left\{
\begin{array}{ll}
 \big(A + B K^0_*\big)\hat X_t& \text{ if }Z_{t+1}= \emptyset,\\
 X_{t+1} & \text{ if }Z_{t+1} = X_{t+1}.
\end{array}\right.
\label{eq:estimator_inf_2C}
\end{align}
The above strategies are steady state versions of  \eqref{eq:opt_U_2C} - \eqref{eq:estimator_t_2C} obtained by replacing $K^0_t,K^1_t$ with $K^0_*,K^1_*$.
\item If matrices $P_t^0$ and $P_t^1$ do not converge, is it still possible to find control strategies with finite cost for Problem \ref{problem_infinite_2C}?
\end{enumerate}
We answer the above three questions in the following subsections.

\subsection{Answering Q1}
\label{sec:Q1_2C} 
In Q1, we want to know whether $P_t^0$ and $P_t^1$ defined by coupled recursions of \eqref{eq:P_initial}-\eqref{eq:tildeP_finite_2C} converge to $P_*^0$ and $P_*^1$ satisfying \eqref{eq:P_finite_2C_fixed}-\eqref{eq:tildeP_finite_2C_fixed}.  Our approach for answering Q1 is based on establishing a connection between the recursions for matrices $P_t^0$ and $P_t^1$ in our DNCS problem and the recursions for matrices $P_t^{\diamond}(m)$, $m \in \mathcal{M},$ in the MJLS problem reviewed in Section \ref{sec:mjls}. This approach consists of the following two steps.



\subsubsection*{\textbf{Step 1: Constructing an auxiliary MJLS}}
Consider an auxiliary MJLS where the set $\mathcal{M}$ of modes is $\{0,1\}$. Then, we have the following two sequences of  matrices, $P^{\diamond}_t(0),P^{\diamond}_t(1),$ defined recursively using \eqref{eq:CARE_init} and \eqref{eq:CARE_finite} for this MJLS:
\begin{align}
&P^{\diamond}_{T+1}(0) = P^{\diamond}_{T+1}(1) =\mathbf{0}, \label{eq:P_MJ_init}\\
&P^{\diamond}_t(0) =  \notag \\
&\Omega\big(\theta^{00} P^{\diamond}_{t+1}(0) + \theta^{01} P^{\diamond}_{t+1}(1),Q^{\diamond}(0),R^{\diamond}(0),A^{\diamond}(0),B^{\diamond}(0) \big),
\label{P_MJ_cmp_2C_0}
\\
&P^{\diamond}_t(1)=   \notag \\
& \Omega\big(\theta^{10} P^{\diamond}_{t+1}(0) + \theta^{11} P^{\diamond}_{t+1}(1),Q^{\diamond}(1),R^{\diamond}(1),A^{\diamond}(1),B^{\diamond}(1) \big).
\label{P_MJ_cmp_2C_1}
\end{align}
Furthermore, recall that from \eqref{eq:P_initial}-\eqref{eq:tildeP_finite_2C}, we have the following recursions for matrices $P_t^0$ and $P_t^1$ in our DNCS problem,
\begin{align}
&P^0_{T+1}=P^1_{T+1}=0, \label{eq:barP_cmp_init}\\
&P_t^0 =  \Omega(P_{t+1}^0,Q,R,A,B),
\label{eq:barP_cmp_0_2C}
\\
&P_t^{1} =  \Omega((1-p^1) P_{t+1}^0+p^1 P_{t+1}^1,Q,R^{11},A,B^{11}).
\label{eq:barP_cmp_n_2C}
\end{align}
Is it possible to find matrices $A^{\diamond}(m),B^{\diamond}(m), Q^{\diamond}(m),R^{\diamond}(m)$, $m \in \{0,1\}$, and a transition probability matrix $\Theta$ for the auxiliary MJLS such that the recursions in \eqref{eq:P_MJ_init} - \eqref{P_MJ_cmp_2C_1} coincide with the recursions in \eqref{eq:barP_cmp_init} - \eqref{eq:barP_cmp_n_2C}?


By comparing \eqref{P_MJ_cmp_2C_0}-\eqref{P_MJ_cmp_2C_1} with \eqref{eq:barP_cmp_0_2C}-\eqref{eq:barP_cmp_n_2C}, we find that the following definitions would make the two sets of equations identical:
\begin{align}
&A^{\diamond}(0) = A^{\diamond}(1) = A, \label{eq:MJLS_A} \\
&B^{\diamond}(0) = B, \quad B^{\diamond}(1) = [\mathbf 0,B^{11}], \label{eq:MJLS_B}
\\
&Q^{\diamond}(0) = Q^{\diamond}(1) = Q,\label{eq:MJLS_Q} \\
&R^{\diamond}(0) = R, \quad R^{\diamond}(1) = \begin{bmatrix}
\mathbf I & \mathbf 0 \\
\mathbf 0 & R^{11}
\end{bmatrix}, \label{eq:MJLS_R}
\\
&\Theta = \begin{bmatrix}
\theta^{00} & \theta^{01} \\
\theta^{10} & \theta^{11}
\end{bmatrix} = \begin{bmatrix}
1 & 0 \\
1-p^1 & p^1
\end{bmatrix}.
\label{eq:MJLS_theta}
\end{align}
To complete the definition of the auxiliary MJLS, we need to define the initial state and mode probability distributions $\pi_{X_0^{\diamond}}$ and $\pi_{M_0}$. These can be defined arbitrarily and for simplicity we assume that the initial state and mode of the auxiliary MJLS are fixed to be $X^{\diamond}_0 = 0$ and $M_0 = 1$.
The following lemma summarizes the above discussion.


\begin{lemma}
\label{equality_recursions_2C}
For the auxiliary MJLS described by \eqref{eq:MJLS_A}-\eqref{eq:MJLS_theta}, the coupled recursions in \eqref{eq:P_MJ_init}-\eqref{P_MJ_cmp_2C_1} are identical to  the coupled recursions in \eqref{eq:barP_cmp_init}-\eqref{eq:barP_cmp_n_2C}.
\end{lemma}
\begin{proof}
The lemma can be proved by  straightforward algebraic manipulations.
\end{proof}

\begin{remark}
Note that  we have not defined $B^{\diamond}(1)$ to be $B^{11}$  because the MJLS model requires that the dimensions of matrices $B^{\diamond}(0)$ and $B^{\diamond}(1)$ be the same (see Section \ref{sec:mjls}). Similar dimensional considerations  prevent us from defining $R^{\diamond}(1)$ to be simply $R^{11}$.
\end{remark}

\begin{remark}\label{rem:fictitious}
It should be noted that the auxiliary MJLS is simply a mathematical device. It cannot be seen as a reformulation or  another interpretation of our DNCS problem. In particular, the binary mode $M_t$  \emph{is not} the same as the link state $\Gamma^1_t$. The distinction between $M_t$ and $\Gamma^1_t$ is immediately clear if one recalls that $M_t$ is the state of a Markov chain with transition probability matrix given in \eqref{eq:MJLS_theta} whereas $\Gamma^1_t, t \geq 0,$ is an i.i.d process. 
\end{remark}
\subsubsection*{\textbf{Step 2: Using MJLS results to answer Q1}}
Now that we have constructed an auxiliary MJLS such that $P^{\diamond}_t(m) = P_t^{m}$ for $m \in \{0,1\}$, we can use the MJLS results about convergence of  matrices $P^{\diamond}_t(m)$ (that is, Lemmas \ref{lm:MJ_infinite} and \ref{lm:ss})  to answer Q1. The following lemma states this result.


\begin{lemma}\label{lm:pc_2C}
Suppose Assumption \ref{assum:det_stb_2C} holds. Then, the matrices $P_t^0$ and $P_t^1$ defined in \eqref{eq:P_initial}-\eqref{eq:tildeP_finite_2C} converge as $t \to -\infty$ to matrices $P_*^0$ and $P_*^1$ that satisfy the coupled fixed point  equations \eqref{eq:P_finite_2C_fixed} - \eqref{eq:tildeP_finite_2C_fixed} if and only if $p^1 < p^1_c$, where $p^1_c$ is the critical threshold given by 
\begin{align}
\frac{1}{\sqrt{p^1_c}} =  \min_{K  \in \mathbb{R}^{d^1_U \times d_X}}\rho(A+B^{11}K).
\label{eq:pc_2c}
\end{align} 
%
\end{lemma}

\begin{proof}
See Appendix \ref{proof_lm:pc_2C}.
\end{proof}

\subsection{Answering Q2 and Q3}
\label{sec:Q2_2C}
Assuming that $P_{t}^n \rightarrow P_*^n$ as $t \rightarrow -\infty$ for $n=0,1,$ we want to know whether the control strategies of \eqref{eq:K_finite_2C_fixed}-\eqref{eq:estimator_inf_2C} are optimal for Problem \ref{problem_infinite_2C}. The following result shows that these control strategies are indeed optimal.

\begin{lemma}
\label{lm:Q2_2C}
If $P_{t}^n \rightarrow P_*^n$ as $t \rightarrow -\infty$ for $n=0,1$, then
\begin{enumerate}
\item Problem \ref{problem_infinite_2C} has finite optimal cost,
\item The strategies described by \eqref{eq:K_finite_2C_fixed}-\eqref{eq:estimator_inf_2C} are optimal for Problem \ref{problem_infinite_2C},
\item Under the strategies described by \eqref{eq:K_finite_2C_fixed}-\eqref{eq:estimator_inf_2C}, $X_t$ and $(X_t-\hat X_t)$ are mean square stable, i.e.,
\[ \sup_{t \geq 0} \mathds{E}^{g^*}[||X_t||^2]  < \infty \mbox{~and~} \sup_{t \geq 0} \mathds{E}^{g^*}[||(X_t-\hat X_t)||^2]  < \infty,\]
 where $g^*$ denotes the strategy described by \eqref{eq:K_finite_2C_fixed}-\eqref{eq:estimator_inf_2C}.
\end{enumerate}
\end{lemma}
\begin{proof}
See Appendix \ref{sec:lm_Q2_2C} for proof of parts 1) and 2). See Appendix \ref{sec:stability_proof} for proof of part 3).
\end{proof}


The following lemma answers Q3.
\begin{lemma}
\label{lm:Q3}
If  $P_t^n$, $n=0,1,$ do not converge as $t \to -\infty$, then Problem \ref{problem_infinite_2C} does not have finite optimal cost. 
\end{lemma}
\begin{proof} See Appendix \ref{sec:lmQ3}.
\end{proof}

Now that we have answered  Q1, Q2 and Q3, we can summarize our results for the infinite horizon DNCS problem (Problem \ref{problem_infinite_2C}).

\subsection{Summary of the Infinite Horizon Results}
Based on the answers to Q1-Q3, the following theorem summarizes our results for Problem \ref{problem_infinite_2C}.
\begin{theorem}
\label{thm:DC_infinite_2C}
Suppose Assumption \ref{assum:det_stb_2C} holds. Then, 
\begin{enumerate}[(i)]
\item Problem \ref{problem_infinite_2C} has finite optimal cost if and only if $p^1 < p^1_c$ where the critical threshold $p^1_c$ is given by \eqref{eq:pc_2c}.

\item If $p^1 < p^1_c$, there exist symmetric positive semi-definite matrices $P_*^0, P_*^1$ that satisfy \eqref{eq:P_finite_2C_fixed}-\eqref{eq:tildeP_finite_2C_fixed} and the optimal strategies  for Problem \ref{problem_infinite_2C} are given by
\begin{align}
\bmat{U^{0*}_t \\ U^{1*}_t \\} = K_*^0\hat X_t + \bmat{\mathbf{0} \\ K_*^1 }\left(X_t - \hat X_t\right),
\label{eq:opt_U_infinite_2C_thm}
\end{align}
where the  estimate $\hat X_t$ can be computed recursively using \eqref{eq:estimator_inf_2C} with $\hat{X}_0=0$
 and the gain matrices $K_*^0,K_*^1$ are given by
\begin{align}
& K_*^0 = \Psi(P_{*}^0,R,A,B),
\label{eq:K_finite_2C_fixed_final}
\\
& K_*^1 = \Psi((1-p^1)P_{*}^0+p^1 P_{*}^1,R^{11},A,B^{11}).
\label{eq:tildeK_finite_2C_fixed_final}
\end{align} 
\item If $p^1 < p^1_c$, then under the strategies described in part (ii) above, $X_t$ and $(X_t-\hat X_t)$ are mean square stable.
 

\end{enumerate}

\end{theorem}
\begin{proof}
The result follows from Lemmas \ref{lm:pc_2C}, \ref{lm:Q2_2C} and \ref{lm:Q3}.
\end{proof}



If $B^{11} = 0$, the local controller becomes just a sensor without any control ability. In this case, Theorem \ref{thm:DC_infinite_2C} gives the 
critical threshold as $p^1_c = \rho(A)^{-2}$ and the closed-loop system is mean-square stable if $\rho(A) < 1/\sqrt{p^1}$. This recovers the single-controller NCS result in \cite{Imer2006optimal}. Thus, we have the following corollary of Theorem \ref{thm:DC_infinite_2C}.
\begin{corollary}[Theorem 3 of \cite{Imer2006optimal} with $\alpha = 0$ and $\beta = p^1$]
Suppose the local controller is just a sensor (i.e., $B^{11} = 0$) and the remote controller is the only controller present. Then, if $\rho(A) < 1/\sqrt{p^1}$, the optimal controller of this single-controller NCS is given by $U^{0*}_t$ in \eqref{eq:opt_U_infinite_2C_thm}, and the corresponding closed-loop system is mean-square stable.
\end{corollary}

\begin{remark}
The value $\frac{1}{\sqrt{p^1_c}} =  \min_{K}\rho(A+B^{11}K)$ is actually the largest unreachable mode of $(A,B^{11})$. Therefore, it can be computed using tests for reachability such as the Popov-Belovich-Hautus (PBH) test \cite{hespanha2009linear}. Further, if  $(A,B^{11})$ is reachable, then $\rho(A+B^{11}K)$ can be made arbitrarily small which implies that $p^1_c = \infty$. This is the case when the local controller can stabilize the system by itself, so the DNCS is stabilizable under any link failure probability $p^1$.
\end{remark}

\begin{remark}
If $p^1 < p^1_c$, the coupled Riccati equations in \eqref{eq:P_finite_2C_fixed}-\eqref{eq:tildeP_finite_2C_fixed} can be solved by iteratively carrying out the recursions in \eqref{eq:P_initial}-\eqref{eq:tildeP_finite_2C} until convergence. This is similar to the procedure in \cite[Chapter 7]{costa2006discrete}.
\end{remark}

\section{Extension to Multiple Local Controllers}
\label{sec:model_N_controllers}
%

In this section, we study an extension of the system model in Section \ref{sec:model_2_controllers} to the case where instead of 1 local controller, we have $N$ local controllers, $C^1,C^2, \ldots,C^N$, each associated to a co-located plant as shown in Fig. \ref{fig:SystemModel}. We use $\mathcal{N}$ to denote the set $\{1,2, \ldots, N\}$ and $\overline{\mathcal{N}}$ to denote $\{0,1,\ldots, N\}$.
The linear dynamics of plant $n \in \mathcal{N}$ are given by
\begin{align}
&X_{t+1}^n \!=\! A^{nn} X_t^n + B^{nn}U^{n}_t+ B^{n0} U^0_t + W_t^n, t=0,\dots,T,
 \label{Model:system}
\end{align}
where $X_t^n\in \R^{d_X^n}$ is the state of the plant $n$ at time $t$,
$U^n_t \in \R^{d_U^n}$ is the control action of the controller $C^{n}$,  $U^0_t \in \R^{d_U^0}$ is the control action of the controller $C^{0}$, and $A^{nn}, B^{nn}, B^{n0}$ are matrices with appropriate dimensions.
We assume that $X^n_0 = 0$, and that  $W^n_t$, $n \in \mathcal{N}, t \geq 0$, are i.i.d  random variables with zero mean and $\cov(W_t^n) = \mathbf{I}$. Note that we do not assume that random variables $W_{t}^{n}$, $n \in \mathcal{N}, t \geq 0$, are Gaussian.

The overall system  dynamics can be written as
\begin{align}
X_{t+1} = A X_t + BU_t + W_t,
\label{overall_state_dynamic}
\end{align}
where $X_t = \vecc(X^{1:N}_t), U_t = \vecc(U^{0:N}_t),W_t = \vecc(W^{0:N}_t)$ and $A,B$ are defined as

\begin{align}
A &= \begin{bmatrix}
   A^{11} & & \text{\huge0}\\
          & \ddots & \\
     \text{\huge0} & & A^{NN}
\end{bmatrix}, 
B= 
\begin{bmatrix}
B^{10} &  B^{11} & & \text{\huge0}\\
\vdots     &     & \ddots & \\
B^{N0}  &   \text{\huge0} & & B^{NN}
\end{bmatrix}.
\label{eq:thm_matricesABB}
\end{align}

\begin{singlespace}

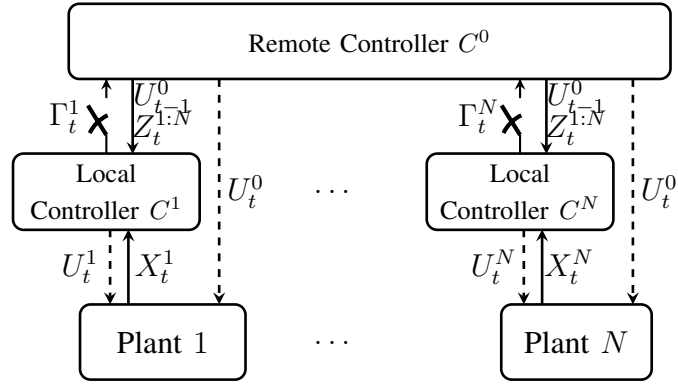
\begin{figure}
\begin{center}
\begin{tikzpicture}[every text node part/.style={align=center}]
\begin{scope}[rotate=180]
\node [rectangle,draw,minimum width=8cm,minimum height=1cm,line width=1pt,rounded corners]at (2,-2) (1) {{\small Remote Controller $C^0$}}; 

\begin{scope}[shift={(0.25,0)}]
\node [rectangle,draw,minimum width=2.2cm,minimum height=1cm,line width=1pt,rounded corners,text width=2.2cm]at (-0.25,0) (2) {{\small Local Controller $C^{N}$}}; 
\node [rectangle,draw,minimum width=2cm,minimum height=1cm,line width=1pt,rounded corners]at (-1,2) (3) {Plant $N$}; 

\path[thick,->,>=stealth,line width=1pt,dashed]
           (-0.3,0.5) edge node {}   (-0.3,1.5) 
           (-1.75,-1.5) edge node {}   (-1.75,1.5) 
     ;
      
      \path[thick,->,>=stealth,line width=1pt]
           (-0.6,-1.5) edge node {}   (-0.6,-0.5) 
      ;
      
    \path[thick,->,>=stealth, shift left=.30ex,line width=1pt]
        (-0.6,1.5) edge node {}   (-0.6,0.5); 

\node[] at (.1,1) {$U_t^{N}$};
\node[] at (-0.9,1) {$X^N_t$};
\node[] at (-2.1,0) {$U_t^0$};
\node[] at (.3,-1) {$\Gamma_t^N$};
\node[] at (-1,-.85) {$Z_t^{1:N}$};
\node[] at (-1,-1.25) {$U_{t-1}^0$};
\end{scope}

\begin{scope}[shift={(-0.25,0)}]
\node [rectangle,draw,minimum width=2.2cm,minimum height=1cm,line width=1pt,rounded corners,text width=2.2cm]at (5.75,0) (2) {{\small Local Controller $C^{1}$}}; 
\node [rectangle,draw,minimum width=2.2cm,minimum height=1cm,line width=1pt,rounded corners]at (5,2) (4) {Plant $1$}; 

\path[thick,->,>=stealth,line width=1pt,dashed]
           (5.7,0.5) edge node {}   (5.7,1.5)           
           (4.25,-1.5) edge node {}   (4.25,1.5) 
      ;
      
      \path[thick,->,>=stealth,line width=1pt]
           (5.4,-1.5) edge node {}   (5.4,-.5)      
      ;
      
    \path[thick,->,>=stealth, shift left=.30ex,line width=1pt]
        (5.4,1.5) edge node {}   (5.4,0.5) ; 

\node[] at (6.1,1) {$U_t^{1}$};
\node[] at (5.1,1) {$X^1_t$};
\node[] at (3.9,0) {$U_t^0$};
\node[] at (6.3,-1) {$\Gamma_t^1$};
\node[] at (5.0,-.85) {$Z_t^{1:N}$};
\node[] at (5.0,-1.25) {$U_{t-1}^0$};
\end{scope}

\node[] at (2.5,2) {$\ldots$};
\node[] at (2.5,0) {$\ldots$};
\end{scope}
\draw [line width=0.8pt,addarrow] (-5.5,0.5) to[cspst] (-5.5,1.5); 
\draw [line width=0.8pt,addarrow] (0,0.5) to[cspst] (0,1.5); 
%
\end{tikzpicture}
\caption{System model. The binary random variables $\Gamma_t^{1:N}$ indicate whether packets are transmitted successfully. Dashed lines indicate control links and solid lines indicate communication links.}
\label{fig:SystemModel}
\end{center}
\end{figure}
\end{singlespace}

\subsubsection*{Communication Model}
\label{subs:comm_model}
The communication model is similar to the one described in  Section \ref{subs:comm_model_2C}. In particular, for each $n \in \mathcal{N}$, there is an unreliable link with link failure probability $p^n$ from the local controller $C^{n}$ to the remote controller $C^0$. The local controller $C^{n}$ uses its unreliable  link to send the  state  $X_t^n$ of its co-located plant to the remote controller. The state of this link at time $t$ is described by a Bernoulli random variable $\Gamma_t^n$ and the output of this link at time $t$ is denoted by $Z_t^n$, where $\Gamma_t^n$ and $Z_t^n$ are described by equations  similar to \eqref{eq:gamma} and \eqref{Model:channel_2C}. We assume that $\Gamma_{0:t}^{1:N}$, $t \geq 0$, are independent random variables and that they are independent of $W_{0:t}^{1:N}$, $t \geq 0$.

Unlike the unreliable uplinks, we assume that there exist perfect links from $C^0$ to $C^{n}$, for  each $n \in \mathcal{N}$. Therefore, $C^0$ can share $Z_t^{1:N}$ and $U_{t-1}^0$ with all local controllers $C^{1:N}$.
All controllers select their control actions at time $t$ after observing $Z_t^{1:N}$  and  $U_{t-1}^0$. 
We assume that for each $n \in \mathcal{N}$, the links from controllers $C^{n}$ and $C^0$ to  plant $n$ are perfect.


\subsubsection*{Information structure and cost}
\label{subs:info_cost_2C}
Let $H^{n}_t$ denote the information available to controller $C^{n}$, $n \in {\color{black} \overline{\mathcal{N}}}$,  at time $t$. 
Then,
\begin{align}
H^{n}_t&= \{X^n_{0:t}, U^{n}_{0:t-1}, Z_{0:t}^{1:N}, U^0_{0:t-1}\}, \hspace{2mm} n \in \mathcal{N}, \notag \\
H^0_t &= \{Z_{0:t}^{1:N}, U^0_{0:t-1}\}. 
\label{Model:info}
\end{align}
Let $\mathcal{H}^{n}_t$ be the space of all possible realizations of $H_t^n$.
Then, $C^{n}$'s actions are selected according to
\begin{align}
U^{n}_t &= g^{n}_t(H^{n}_t), \hspace{2mm}  n \in {\color{black} \overline{\mathcal{N}}},
\label{Model:strategy}
\end{align}
where $g^{n}_t:\mathcal{H}^{n}_t{\color{black} \to } \R^{d_U^n}$ is a Borel measurable mapping.
We use $g:=(g^0_{0},g^0_1,\dots,g^1_{0},g^1_1,\dots, g^N_{0},g^N_1,\dots,)$ to collectively denote the control strategies of all $N+1$ controllers.

The instantaneous cost $c_t(X_t, U_t)$ of the system is a quadratic function similar to the one described in \eqref{Model:cost_2C} where $X_t = \vecc(X^{1:N}_t), U_t = \vecc(U^{0:N}_t)$ and 
%
\begin{align}
Q= \begin{bmatrix}
Q^{11} &\ldots &Q^{1N} \\
\vdots & \ddots & \vdots \\
Q^{N1} & \ldots & Q^{NN}
\end{bmatrix}, R= \begin{bmatrix}
R^{00} &R^{01} &\ldots &R^{0N} \\
R^{10} &R^{11} &\ldots &R^{1N} \\
\vdots &\vdots & \ddots & \vdots \\
R^{N0}  & \ldots & \ldots & R^{NN} 
\end{bmatrix}.
\label{matrix_structure}
\end{align}
$Q$ is a symmetric positive semi-definite (PSD) matrix and $R$ is a symmetric positive definite (PD) matrix. 

\subsubsection*{Problem Formulation}
\label{subs:prob_formulaiton}
Let $\mathcal{G}$ denote the set of all possible control strategies of controllers $C^{0},\ldots,C^N$.
The performance of control strategies $g$ over a finite horizon $T$ is measured by $J_T(g)$ defined in \eqref{Model:J_T_2C}.
For the decentralized networked control system (DNCS) described above, we consider the problem of strategy optimization  over finite and infinite time horizons. These two problems are formally defined below.
\begin{problem}
\label{problem_finite}
For the DNCS described above, 
solve the following strategy optimization problem:
\begin{align}
\inf_{g \in\mathcal{G}} 
J_T(g)
\label{Model:obj_finite}
\end{align}
\end{problem}

\begin{problem}
\label{problem_infinite}
For the DNCS described above, 
solve the following strategy optimization problem:
\begin{align}
&\inf_{g \in\mathcal{G}} 
J_{\infty}(g)
:= \inf_{g \in\mathcal{G}} 
\limsup_{T\rightarrow\infty} \frac{1}{T+1} J_{T}(g).
\label{Model:obj_infinite}
\end{align}
\end{problem}

Due to stability issues in the infinite horizon problem, we make the following  assumptions on the system and cost matrices. 
\begin{assumption}
\label{assum:det_stb}
$(A,Q^{1/2})$ is detectable and $(A,B)$ is stabilizable. 
\end{assumption}

\begin{assumption}
\label{assum:det_stb_2}
$\big(A^{nn},(Q^{nn})^{1/2} \big)$ is detectable for all $n \in \mathcal{N}$.
\end{assumption}

The finite horizon strategy optimization problem (Problem \ref{problem_finite}) has been solved in \cite{asghari_ouyang_nayyar_tac_2018}.  We summarize the finite horizon results below. 
\begin{lemma}(\cite[Theorem 2]{asghari_ouyang_nayyar_tac_2018})
\label{lm:opt_strategies}
The optimal control strategies of Problem \ref{problem_finite} are given by
\begin{align}
\bmat{U^{0*}_t \\ U^{1*}_t \\ \vdots \\ U^{N*}_t} = K_t^0 \hat X_t + 
\begin{bmatrix}
\mathbf{0} & \ldots & \mathbf{0} \\
&&\\
K_t^{1} &     &  \text{\huge0}  \\
 &\ddots     &   \\
  \text{\huge0} &  & K_t^{N}
\end{bmatrix} \left(X_t - \hat X_t\right), 
\label{eq:opt_U}
\end{align}
where $\hat X_t = \vecc(\hat X^{1:N}_t) = \vecc(\ee[X^1_t|H^0_t], \ldots, \ee[X^N_t|H^0_t]) = \ee[X_t|H^0_t]$ is the  estimate (conditional expectation) of $X_t$ based on the common information $H^0_t$.  The estimate $\hat X_t$ can be computed recursively as follows: for $n \in \mathcal{N}$,
\begin{align}
\hat X_0^n = & 0, 
\label{eq:estimator_0}
\\
\hat X_{t+1}^n
= &\left\{
\begin{array}{ll}
 \big([A]_{n,:} + [B]_{n,:} K^0_t\big)\hat X_t& \text{ if }Z_{t+1}^n= \emptyset,\\
 X_{t+1}^n & \text{ if }Z_{t+1}^n = X_{t+1}^n.
\end{array}\right. 
\label{eq:estimator_t}
\end{align}
The gain matrices $K_t^0$ and $K_t^{n}$, $n \in \mathcal{N}$, are given by
\begin{align}
& K^0_t = \Psi(P_{t+1},R,A,B),
\label{eq:K_finite}
\\
&K_t^n = \Psi \big ((1-p^n)[P_{t+1}^0]_{n,n}+p^n P_{t+1}^{n},R^{nn},A^{nn},B^{nn} \big),
\label{eq:tildeK_finite}
\end{align}
where $P^0_t = \begin{bmatrix}
[P_{t}^0]_{1,1} &\ldots &[P_{t}^0]_{1,N} \\
\vdots & \ddots & \vdots \\
[P_{t}^0]_{N,1} & \ldots & [P_{t}^0]_{N,N}
\end{bmatrix} \in \mathbb{R}^{(\sum_{n=1}^N d^n_X)\times (\sum_{n=1}^N d^n_X)}$ and $P_t^{n} \in \mathbb{R}^{ d^n_X\times d^n_X} $, for $n \in \mathcal{N}$, are PSD matrices obtained recursively as follows:  
where $P_t^0$ and $P_t^{n}$,  $n \in \mathcal{N}$, are PSD matrices obtained recursively as follows:
\begin{align}
&P_{T+1}^0 = \mathbf{0}, \quad P_{T+1}^{n} = \mathbf{0}, \label{eq:P_N_init}\\
&P_t^0 =  \Omega(P_{t+1}^0,Q,R,A,B),
\label{eq:P_finite}
\\
&P_t^{n} =  \Omega \big((1-p^n)[P_{t+1}^0]_{n,n}+p^n P_{t+1}^{n},Q^{nn},R^{nn},A^{nn},B^{nn}\big).
\label{eq:tildeP_finite}
\end{align}
Furthermore, the optimal cost is given by
\begin{align}
J^*_{T} = & \sum_{t = 0}^T \sum_{n = 1}^N \tr \Big((1-p^n)[P_{t+1}^0]_{n,n}+p^n P_{t+1}^{n} \Big).
\label{eq:opcost_finite}
\end{align}
\end{lemma}

\begin{remark}
Note that remote controller's action $U^{0*}_t$ in \eqref{eq:opt_U} is a function  of $\hat{X}_t$ only while the local controller $C^n$'s action $U^{n*}_t$ is a function of both $\hat{X}_t$ and $X_t^n$. Further, as per \eqref{eq:estimator_0} and \eqref{eq:estimator_t}, $\hat{X}_t$ is computed recursively based only on the knowledge of $Z_{0:t}^{1:N}$.
\end{remark}

\subsection{Infinite Horizon Optimal Control}
\label{sec:infinite_N_controllers}
As in Section \ref{sec:infinite_2_controllers},  the infinite horizon  problem (Problem \ref{problem_infinite}) can be solved by answering the following three questions:

\begin{enumerate}[label=(\mylabel{Q}{\arabic*})]
\item Do matrices $P_t^{0},\ldots,P_t^N$, defined in \eqref{eq:P_N_init}-\eqref{eq:tildeP_finite} converge as $t \to -\infty$ to $P_*^0,\ldots, P_*^N$, that satisfy the coupled fixed point  equations \eqref{eq:P_finite_NC_fixed} - \eqref{eq:tildeP_finite_NC_fixed} below?
\begin{align}
&P_*^0 =  \Omega \big(P_{*}^0,Q,R,A,B\big),
\label{eq:P_finite_NC_fixed}
\\
& P_*^{n} =  \Omega \big((1-p^n)[P_{*}^0]_{n,n}+p^n P_{*}^{n},Q^{nn},R^{nn},A^{nn},B^{nn}\big).
\label{eq:tildeP_finite_NC_fixed}
\end{align}
\item If matrices $P_t^{0},\ldots,P_t^N$ converge and we define matrices $K_*^0, \ldots, K_*^N$, using matrices $P_*^{0},\ldots,P_*^N$ as follows, 
\begin{align}
& K_*^0 = \Psi(P_{*}^0,R,A,B),
\label{eq:K_finite_NC_fixed}
\\
& K_*^n = \Psi \big((1-p^n)[P_{*}^0]_{n,n}+p^n P_{*}^n,R^{nn},A^{nn},B^{nn}\big),
\label{eq:tildeK_finite_NC_fixed}
\end{align}
are the following strategies optimal for Problem \ref{problem_infinite}?
\begin{align}
\bmat{U^{0*}_t \\ U^{1*}_t \\ \vdots \\ U^{N*}_t} = K_*^0 \hat X_t + 
\begin{bmatrix}
\mathbf{0} & \ldots & \mathbf{0} \\
&&\\
K_*^{1} &     &  \text{\huge0}  \\
 &\ddots     &   \\
  \text{\huge0} &  & K_*^{N}
\end{bmatrix} \left(X_t - \hat X_t\right), 
\label{eq:opt_U_NC_fixed}
\end{align}
where $\hat{X}_t$ can be computed recursively using \eqref{eq:estimator_0} - \eqref{eq:estimator_t}  by replacing $K^0_t$ with $K^0_*$.

\item If matrices $P_t^{0},\ldots,P_t^N$ do not converge, is it still possible to find control strategies with finite cost for Problem \ref{problem_infinite}?
\end{enumerate}


\subsection{Answering Q1}
As in Section \ref{sec:Q1_2C}, we will answer Q1 by establishing a connection between the recursions for matrices $P_t^{0},\ldots,P^N_t$ in our DNCS problem and the recursions for matrices $P_t^{\diamond}(m)$, $m \in \mathcal{M},$ in the MJLS problem.  One obstacle in making this connection is the fact that the matrices $P_t^{0},\ldots,P^N_t$ in our DNCS problem do not have the same dimensions while the matrices $P_t^{\diamond}(m)$, $m \in \mathcal{M},$ in the MJLS problem  all have the same dimensions. This obstacle was not present in Section \ref{sec:infinite_2_controllers}. To get around this difficulty, we first provide a new representation $\bar P_t^{0},\ldots,\bar P_t^{N}$ for the matrices $P_t^{0},\ldots,P^N_t$ in our DNCS problem such that the new  matrices $\bar P_t^{0},\ldots,\bar P_t^{N}$ all have the same dimensions.


\begin{lemma}
\label{lm:opt_strategies_new_rep}
Define matrices $\bar P_t^{n} \in \mathbb{R}^{(\sum_{k=1}^N d^k_X)\times (\sum_{k=1}^N d^k_X)}$, $n =0,1,\ldots,N$, recursively as follows:
\begin{align}
&\bar P_{T+1}^n = \mathbf{0}, \label{eq:barP_init}\\
&\bar P_t^0 = \Omega(\bar P_{t+1}^0,Q,R,A,B),
\label{eq:barP_finite_0}
\\
&\bar P_t^{n} = \Omega \big((1-p^n)\bar P_{t+1}^0+p^n \bar P_{t+1}^n, 
 \mathcal{L}_{zero}(Q,Q^{nn},n,n),  \notag \\
& \hspace{1.4cm}  \mathcal{L}_{iden}(R,R^{nn},n+1), \mathcal{L}_{zero}(A,A^{nn},n,n), \notag \\
 & \hspace{1.4cm} \mathcal{L}_{zero}(B,B^{nn},n,n+1) \big),
\label{eq:barP_finite}
\end{align}
where the operators $\mathcal{L}_{zero}$ and $\mathcal{L}_{iden}$ are as defined in Section \ref{sec:operators}.
Then, for $ t \leq T+1$,
\begin{align}
\label{new_rep1}
&\bar P_t^0 =  P_t^0, \\
& \bar P_t^{n} = \mathcal{L}_{zero}(P^0_t, P_t^{n},n,n), ~~n=1,\ldots,N.
\label{new_rep2}
\end{align}
Consequently, matrices $P_t^{0},\ldots,P^N_t$ converge as $t \to -\infty$ if and only if matrices $\bar P_t^{0},\ldots,\bar P_t^{N}$ converge as $t \to -\infty$.
\end{lemma}
\begin{proof}
See Appendix \ref{proof_lm:opt_strategies_new_rep} for a proof.
\end{proof}

%
We can now proceed with constructing an auxiliary MJLS.
\subsubsection*{\textbf{Step 1: Constructing an auxiliary MJLS}}
Consider an auxiliary MJLS where the set $\mathcal{M}$ of modes is $ \{0,1,\ldots,N\}$.  Then, we have the following $N+1$ sequences of  matrices, $P^{\diamond}_t(0),P^{\diamond}_t(1),\ldots,P^{\diamond}_t(N)$, defined recursively using \eqref{eq:CARE_init} and \eqref{eq:CARE_finite} for this MJLS:
\begin{align}
&P^{\diamond}_{T+1}(m)  =\mathbf{0}, ~\forall m \in \mathcal{M}, \label{eq:P_N_MJ_init}\\
&P^{\diamond}_t(0) =  \notag \\
&\Omega\big(\sum_{k=0}^N \theta^{0k} P^{\diamond}_{t+1}(k),Q^{\diamond}(0),R^{\diamond}(0),A^{\diamond}(0),B^{\diamond}(0) \big),
\label{P_MJ_cmp_NC_0}
\\
&P^{\diamond}_t(n)=   \notag \\
& \Omega\big(
\sum_{k=0}^N \theta^{nk} P^{\diamond}_{t+1}(k),Q^{\diamond}(n),R^{\diamond}(n),A^{\diamond}(n),B^{\diamond}(n) \big).
\label{P_MJ_cmp_NC_1}
\end{align}
Furthermore, we have the recursions of \eqref{eq:barP_init}-\eqref{eq:barP_finite} for matrices $\bar P_t^{0},\ldots, \bar P_t^N$ in our DNCS problem.
%
By comparing \eqref{P_MJ_cmp_NC_0}-\eqref{P_MJ_cmp_NC_1} with \eqref{eq:barP_finite_0}-\eqref{eq:barP_finite}, we find that the following definitions would make the two sets of equations identical:
\begin{align}
&A^{\diamond}(0) = A, A^{\diamond}(n) = \mathcal{L}_{zero}(A,A^{nn},n,n), \hspace{1mm} n \in \mathcal{N},
\label{A_mj}
\end{align}
\begin{align}
&B^{\diamond}(0) = B, B^{\diamond}(n) =  \mathcal{L}_{zero}(B,B^{nn},n,n+1), n \in \mathcal{N},
\label{B_mj}
\end{align}
\begin{align}
&Q^{\diamond}(0) = Q, Q^{\diamond}(n) =   \mathcal{L}_{zero}(Q,Q^{nn},n,n), \hspace{1mm} n \in \mathcal{N},
\label{Q_mj}
\\
&R^{\diamond}(0) = R, R^{\diamond}(n) =  \mathcal{L}_{iden}(R,R^{nn},n+1), \hspace{1mm} n \in \mathcal{N},
\label{R_mj} 
\\
&\Theta =
\begin{blockarray}{rccccc}
&0 &1 &2  &\ldots &N  \\
\begin{block}{r[ccccc]}
                     0 & 1 & 0 & \ldots &\ldots & 0 \\
                     1 & 1-p^1 & p^1 &\ddots  & & \vdots \\
                     2 & 1-p^2 &0  &p^2  & \ddots& \vdots \\
                     \vdots & \vdots & \vdots  &\ddots  &\ddots & 0 \\
                     N & 1-p^N & 0  &\ldots  &0 &p^N \\
\end{block}
\end{blockarray}.
\label{transition_prob}
\end{align}
To complete the definition of the auxiliary MJLS, we need to define the initial state and mode probability distributions $\pi_{X_0^{\diamond}}$ and $\pi_{M_0}$. These can be defined arbitrarily and for simplicity we assume that the initial state is fixed to be $X^{\diamond}_0 = 0$ and the initial mode $M_0$ is uniformly distributed over the set $\mathcal{M}$.  The following lemma summarizes the above discussion.


\begin{lemma}
\label{equality_recursions_NC}
For the auxiliary MJLS described by \eqref{A_mj}-\eqref{transition_prob}, 
 the coupled recursions in \eqref{eq:P_N_MJ_init}-\eqref{P_MJ_cmp_NC_1} are identical to  the coupled recursions in \eqref{eq:barP_init}-\eqref{eq:barP_finite}.
\end{lemma}
\begin{proof}
The lemma can be proved by  straightforward algebraic manipulations.
\end{proof}

\subsubsection*{\textbf{Step 2: Using MJLS results to answer Q1}}
Now that we have constructed an auxiliary MJLS where $P^{\diamond}_t(m) = \bar P_t^{m}$ for $m=0,\ldots,N$, we can use the MJLS results about convergence of  matrices $P^{\diamond}_t(m)$ (that is, Lemmas \ref{lm:MJ_infinite} and \ref{lm:ss})  to answer Q1. The following lemma states this result.

\begin{lemma}
\label{lm:pc_NC}
Suppose Assumptions \ref{assum:det_stb} and \ref{assum:det_stb_2} hold. Then, the matrices $ P^0_t,\ldots,  P^N_t$ defined in \eqref{eq:P_N_init}-\eqref{eq:tildeP_finite} converge as $t \to -\infty$ to matrices  $P_*^0,\ldots, P_*^N$ that satisfy the coupled fixed point  equations \eqref{eq:P_finite_NC_fixed}-\eqref{eq:tildeP_finite_NC_fixed} if and only if  $p^n < p_c^n$ for all $n \in \mathcal{N}$, where the critical thresholds $p_c^n, n \in \mathcal{N}$, are given by
\begin{align}
\frac{1}{\sqrt{p_c^n}} =  \min_{K  \in \mathbb{R}^{d^n_U \times d^n_X}}\rho(A^{nn}+B^{nn}K).
\label{eq:pc}
\end{align}
 
\end{lemma}
\begin{proof}
See Appendix \ref{proof_lm:pc}.
\end{proof}

\subsection{Answering Q2 and Q3}
Assuming that $P_{t}^n \rightarrow P_*^n$ as $t \rightarrow -\infty$ for $n=0,\ldots,N,$ we want to know whether the control strategies of \eqref{eq:K_finite_NC_fixed}-\eqref{eq:opt_U_NC_fixed} are optimal for Problem \ref{problem_infinite_2C}. The following result shows that these control strategies are indeed optimal.


\begin{lemma}
\label{lm:Q2_NC}
If $P_{t}^n \rightarrow P_*^n$ as $t \rightarrow -\infty$ for $n=0,\ldots,N$, then
\begin{enumerate}
\item Problem \ref{problem_infinite} has finite optimal cost,
\item The strategies described by \eqref{eq:K_finite_NC_fixed}-\eqref{eq:opt_U_NC_fixed} are optimal for Problem \ref{problem_infinite},
\item Under the strategies described by \eqref{eq:K_finite_NC_fixed}-\eqref{eq:opt_U_NC_fixed}, $X_t$ and $(X_t-\hat X_t)$ are mean square stable.
\end{enumerate}
\end{lemma}
\begin{proof}
See Appendix \ref{sec:proof_lm_Q2_NC}.
\end{proof}
The following lemma answers Q3.
\begin{lemma}\label{lm:Q3_NC}
If matrices $P_t^{0},\ldots,P_t^{N}$ do not converge as $t \rightarrow -\infty$, then Problem \ref{problem_infinite} does not have finite optimal cost. 
\end{lemma}
\begin{proof}
See Appendix \ref{sec:proof_lm_Q3_NC}.
\end{proof}

Now that we have answered  Q1, Q2 and Q3, we can summarize our results for the infinite horizon DNCS problem (Problem \ref{problem_infinite}).

\subsection{Summary of the Infinite Horizon Results}
Based on the answers to Q1-Q3, the following theorem summarizes our results for Problem   \ref{problem_infinite}.
\begin{theorem}
\label{thm:DC_infinite_NC}
Suppose Assumptions \ref{assum:det_stb} and \ref{assum:det_stb_2} hold. Then, 
\begin{enumerate}[(i)]
\item Problem \ref{problem_infinite} has finite optimal cost if and only if
for all $n \in \mathcal{N}$, $p^n < p_c^n$ where the critical threshold $p_c^n$ is given by \eqref{eq:pc}.

\item If $p^n < p_c^n$ for all $n \in \mathcal{N}$, there exist symmetric positive semi-definite matrices $P_*^0,\ldots, P_*^N$ that satisfy \eqref{eq:P_finite_NC_fixed}-\eqref{eq:tildeP_finite_NC_fixed} and the optimal strategies for Problem \ref{problem_infinite} are given by
\begin{align}
\bmat{U^{0*}_t \\ U^{1*}_t \\ \vdots \\ U^{N*}_t} = K_*^0 \hat X_t + 
\begin{bmatrix}
\mathbf{0} & \ldots & \mathbf{0} \\
&&\\
K_*^{1} &     &  \text{\huge0}  \\
 &\ddots     &   \\
  \text{\huge0} &  & K_*^{N}
\end{bmatrix} \left(X_t - \hat X_t\right), 
\label{eq:opt_U_NC_thm}
\end{align}
where the  estimate $\hat X_t$ can be computed recursively using \eqref{eq:estimator_0}-\eqref{eq:estimator_t} by replacing $K^0_t$ with $K^0_*$ and the gain matrices $K_*^0, \ldots,K_*^N$ are given by 
\begin{align}
& K_*^0 = \Psi(P_{*}^0,R,A,B),
\label{eq:K_finite_NC_final}
\\
& K_*^n = \Psi \big((1-p^n)[P_{*}^0]_{n,n}+p^n P_{*}^n,R^{nn},A^{nn},B^{nn}\big).
\label{eq:tildeK_finite_NC_final}
\end{align}
\item If $p^n < p_c^n$ for all $n \in \mathcal{N}$, then under the strategies described in part (ii) above, $X_t$ and $(X_t-\hat X_t)$ are mean square stable.
\end{enumerate}
\end{theorem}

\begin{corollary}
Suppose the local controllers are just  sensors (i.e., $B^{nn} = 0$ for $n=1,\ldots,N$) and the remote controller is the only controller present. Then, if $\rho(A) < 1/\sqrt{p^n}$ for all $n=1,\ldots,N$, the optimal controller of this multi-sensor,  single-controller NCS is given by $U^{0*}_t$ in \eqref{eq:opt_U_NC_thm}, and the corresponding closed-loop system is mean-square stable.
\end{corollary}

\begin{remark}
If Assumption \ref{assum:det_stb_2} is not true, define, for $n=1,\ldots,N$, $p^n_c = \min \{p^n_s, p^n_d\},$  where 
\begin{align}
\frac{1}{\sqrt{p_s^n}} &=  \min_{K  \in \mathbb{R}^{d^n_U \times d^n_X}}\rho(A^{nn}+B^{nn}K).
\label{eq:pc_ss}
\\
\frac{1}{\sqrt{p_d^n}} &=  \min_{H \in \mathbb{R}^{d^n_X \times d^n_X}}\rho \big(A^{nn}+H (Q^{nn})^{1/2} \big).
\label{eq:pc_sd}
\end{align}
Then, using arguments similar to those used for proving Theorem \ref{thm:DC_infinite_NC}, we can show that if $p^n < p_c^n$ for all $n \in \mathcal{N}$, the strategies in \eqref{eq:opt_U_NC_thm} are optimal for Problem \ref{problem_infinite}. Moreover, Problem \ref{problem_infinite} has finite optimal cost and the system state is mean square stable under optimal strategies.
\end{remark}



\section{Discussion}
\label{sec:discussion}
\subsection{Summary of the Approach}
The analysis in Sections \ref{sec:infinite_2_controllers} and  \ref{sec:model_N_controllers} suggests a general approach   for solving infinite horizon decentralized control/DNCS problems. This can be summarized as follows:
\begin{enumerate}
\item Solve the finite horizon version of the DNCS/decentralized control problem (for instance by using the common information approach \cite{nayyar2013decentralized}). Suppose the optimal strategies  are characterized by matrices $P_t^m, m \in \mathcal M = \{1,2,\dots,M\}$ which satisfy $M$ coupled Riccati recursions 
\begin{align}
P_t^m =  \Omega(\sum_{j} \theta^{mj} P_{t+1}^j,Q^m,R^m,A^m,B^m),
\label{eq:P_finite_general}
\end{align}
for some matrices $Q^m,R^m,A^m,B^m$ and positive numbers $\theta^{mj}$ for $m,j \in\mathcal M$. Note that we can scale the $\theta$'s such that $\sum_{j\in\mathcal M}\theta^{mj} = 1$ by appropriately scaling $A^m$ and $R^m$ for all $m\in\mathcal M$.

\item Construct a $M$-mode auxiliary MJLS with transition probabilities $\theta^{mj}$ and system matrices $Q^m,R^m,A^m,B^m$
so that the Riccati recursions associated with optimal control of the MJLS coincide with the Riccati recursions \eqref{eq:P_finite_general}.

\item Analyze stability criteria of the auxiliary MJLS to find conditions under which the Riccati recursions of the DNCS reach a steady state.

\item Verify  that the decentralized strategies characterized by the steady state DNCS Riccati equations are optimal.
\end{enumerate}

\subsection{The Information Switching}
 Even though the auxiliary MJLS we used in our analysis is an artificial system without apparent physical meaning (see Remark \ref{rem:fictitious}), a general DNCS with random packet drops (or random packet delays) does have some aspects of a switched system. In particular, the information at a controller (e.g, the remote controller in our problem) switches between different patterns  based on the state of the underlying communication network.   The information of the remote controller in Section \ref{sec:model_2_controllers} clearly switches between two patterns: no observation when the packet is dropped, and perfect observation when the transmission is successful. The number of such patterns seems  related to (but not always equal to) the number of modes in the MJLS used to analyze the DNCS. For the two-controller DNCS of Section  \ref{sec:model_2_controllers}, the number of patterns between which the remote controller's information switches is the same as the number of modes in its auxiliary MJLS.  For the $N+1$ controller DNCS in Section \ref{sec:model_N_controllers}, the remote controller's information appears to switch between $2^N$ patterns (depending on the state of the $N$ links) but its auxiliary MJLS has only $N+1$ modes. This difference between the number of information patterns and the number of modes is due to the nature of the plant dynamics which ensure that the remote controller's estimate of the $n$th plant is not affected by the state of the $m$th link if $m\neq n$.

\section{Conclusion}
\label{sec:conclusion}

We considered the infinite horizon optimal control problem of a decentralized networked control system (DNCS) with unreliable communication links. We showed that if the link failure probabilities are below certain critical thresholds, then the solutions of certain coupled Riccati recursions reach a steady state and the corresponding decentralized strategies are optimal. Above these thresholds, we showed that no strategy can achieve finite cost. Our main results in Theorems \ref{thm:DC_infinite_2C} and \ref{thm:DC_infinite_NC} explicitly identify the critical thresholds for the link failure probabilities and the corresponding optimal decentralized strategies when all link failure probabilities are below their thresholds. These results were obtained by exploiting a connection between our DNCS Riccati recursions and the coupled Riccati recursions of an auxiliary Markov jump linear system (MJLS). 


\bibliographystyle{ieeetr}
\bibliography{IEEEabrv,References}

\appendices

\section{Properties of the Operators}
\label{app:lm:costgstar}
\begin{lemma}
\label{lm:KPrelation_0}
Consider matrices $P,Q,R,A,B$ of appropriate dimensions with $P,Q$ being PSD matrices and $R$ being  a PD matrix. Define $\Phi(P,K,Q,R,A,B):=Q + K^\tp RK
+ (A+BK)^\tp P(A+BK)$. Then,
\begin{enumerate}
\item[(i)]
\begin{small}
\begin{align}
\Omega(P,Q,R,A,B) &= \Phi(P,\Psi(P,R,A,B),Q,R,A,B)\notag \\& = \min_K \Phi(P,K,Q,R,A,B).
\label{eq:relation1}
\end{align}
\end{small}
Note that the minimization is in the sense of partial order $\preceq$, that is, the minimum value $\Omega(P,Q,R,A,B) \preceq \Phi(P,K,Q,R,A,B)$ for all $K$.

\item[(ii)] Furthermore, for  PSD matrices $Y_1 $ and $Y_2$ such that $Y_1 \preceq Y_2$, we have
\begin{align}
\Omega(Y_1,Q,R,A,B) \preceq \Omega(Y_2,Q,R,A,B).
\label{eq:relation2}
\end{align}
\end{enumerate}
\end{lemma}
\begin{proof}
The statements in the above lemma can be found in the literature (see, for example, \cite[Chapter 2]{whittle_Optimal_Control}). We provide a proof for completeness.
It can be established by straightforward algebraic manipulations that 

\vspace{-2mm}
\begin{small}
\begin{align}
&\Phi(P,K,Q,R,A,B) = \Omega(P,Q,R,A,B) \notag \\
&+(K - \Psi(P,R,A,B))^\tp\mathcal{R}(K - \Psi(P,R,A,B)), \label{eq:new_riccati_1}
\end{align}
\end{small}
with $\mathcal{R}= R + B^\tp PB$. Then \eqref{eq:new_riccati_1} implies that  $\Phi(P,K,Q,R,A,B)$ is minimized when $K=\Psi(P,R,A,B)$ and the minimum value is $\Omega(P,Q,R,A,B)$.

For PSD matrices $Y_1 $ and $Y_2$ such that $Y_1 \preceq Y_2$, it is straightforward to see that $\Phi(Y_1,K,Q,R,A,B) \preceq \Phi(Y_2,K,Q,R,A,B)$ for any $K$. Hence, 
\begin{align}
 \min_K \Phi(Y_1,K,Q,R,A,B) \preceq  \min_K \Phi(Y_2,K,Q,R,A,B).
 \label{min_K_inequality}
\end{align}
From \eqref{min_K_inequality} and \eqref{eq:relation1}, it follows that \eqref{eq:relation2} is correct.
\end{proof}

\section{Proof of Lemma \ref{lm:MJ_infinite}}
\label{app:lm_MJ_infinite}
If matrices $P^{\diamond}_{t}(m)$, $m \in \mathcal{M}$, converge as $t \to -\infty$ to PSD matrices $P^{\diamond}_{*}(m)$, then by continuity, the collection of PSD matrices $P^{\diamond}_* = \{P^{\diamond}_*(0),\ldots, P^{\diamond}_*(M)\}$ satisfy the DCARE in \eqref{eq:CARE_infinite}. Since the DCARE \eqref{eq:CARE_infinite} has a PSD solution  $P^{\diamond}_* = \{P^{\diamond}_*(0),\ldots, P^{\diamond}_*(M)\}$, then from \cite[Proposition 7]{costa1995discrete} and the SD assumption of the MJLS, it is also a stabilizing solution of the DCARE (\cite[Definition 3]{costa1995discrete} and \cite[Definition 4.4]{costa2006discrete}). Then, the MJLS is SS from the definition of the stabilizing solution.


On the other hand, if the MJLS is SS, under the SD assumption of the MJLS, \cite[Corollary A.16]{costa2006discrete} ensures the existence of a stabilizing solution of the DCARE in \eqref{eq:CARE_infinite}.
The solution is also the unique PSD solution from \cite[Theorem A. 15]{costa2006discrete} (by taking $X=0$ in Theorem A. 15). Then from \cite[Proposition A. 23]{costa2006discrete}, matrices $P^{\diamond}_{t}(m)$, $m \in \mathcal{M}$, converge as $t \to -\infty$ to PSD matrices $P^{\diamond}_{*}(m)$.

\section{Proof of Lemma \ref{lm:pc_2C}}
\label{proof_lm:pc_2C}
Because of Lemma \ref{equality_recursions_2C},  $P^0_t = P^{\diamond}_t(0)$ and  $P^1_t = P^{\diamond}_t(1)$, where matrices $P^{\diamond}_t(0), P^{\diamond}_t(1)$ are defined  by \eqref{eq:P_MJ_init} - \eqref{P_MJ_cmp_2C_1} for the auxiliary MJLS. Thus,  we can focus on the convergence of matrices $P^{\diamond}_t(0), P^{\diamond}_t(1)$.


To investigate the convergence of $P^{\diamond}_t(0), P^{\diamond}_t(1)$,  we first show that 
 the auxiliary MJLS described by \eqref{eq:MJLS_A}-\eqref{eq:MJLS_theta} is SS if and only if $p^1 < p_c^1$ where $p_c^1$ is the critical threshold given by \eqref{eq:pc_2c}. According to Lemma \ref{lm:ss}, the MJLS is SS if and only if there exist matrices $K^{\diamond}(m)$, $m \in \{0,1\}$, such that $\rho(\mathcal{A}_s) < 1$. For the MJLS described by \eqref{eq:MJLS_A}-\eqref{eq:MJLS_theta}, we can find $\mathcal{A}_s$ from \eqref{eq:bigmatrix_ss} as follows
\begin{align}
\mathcal{A}_s = \begin{bmatrix}
A_s(0)\otimes A_s(0) & (1-p^1)A_s(1)\otimes A_s(1) \\
 \mathbf{0} & p^1A_s(1)\otimes A_s(1)
\end{bmatrix},
\label{eq:bigmatrix_As}
\end{align}
where $A_s (m) = A^{\diamond} (m) +  B^{\diamond}(m) K^{\diamond}(m)$, $m \in \{0,1\}$. Since the matrix $\mathcal{A}_{s}$ is upper-triangular, it is Schur stable if and only if all its diagonal blocks are  Schur stable. 

Since $A^{\diamond}(0) = A, B^{\diamond}(0) = B$ and  $(A,B)$ is stabilizable from Assumption \ref{assum:det_stb_2C}, there exists $K^{\diamond}(0)$ such that  $\rho\big(A_s(0)\otimes A_s(0)\big)$, which is equal to $\big(\rho\big(A_s(0)\big)\big)^2$, is less than $1$. Therefore,  the MJLS is SS if and only if $\rho\big(p^1 A_s(1)\otimes A_s(1)\big)<1$ for some $K^{\diamond}(1)$. Note that $\rho\big(p^1 A_s(1)\otimes A_s(1)\big) = p^1 \times \big(\rho\big(A_s(1)\big)\big)^2$. Therefore, the MJLS is SS if and only if $\frac{1}{\sqrt{p^1}} > \rho\big(A_s(1)\big)$ for some $K^{\diamond}(1)$. Since $A^{\diamond}(1)= A$ and $B^{\diamond}(1) = [\mathbf 0,B^{11}]$, it follows then that the MJLS is SS iff 
\[\frac{1}{\sqrt{p^1}} > \rho\big(A + B^{11} \tilde K^{\diamond}(1)\big),\] for some $\tilde K^{\diamond}(1)$. This condition is equivalent to $p^1 < p_c^1$ where $p_c^1$ is the critical threshold given by \eqref{eq:pc_2c}.

Next, we show that the auxiliary MJLS described by \eqref{eq:MJLS_A}-\eqref{eq:MJLS_theta} is SD. To this end, we can follow an argument similar to the one  described above for establishing that the MJLS is SS and use part 2 of Lemma \ref{lm:ss} to show that the MJLS is SD if and only if there exist matrices $H^{\diamond}(0)$ and $H^{\diamond}(1)$ such that $\rho\big(A_d(0)\otimes A_d(0)\big)<1$ and $\rho\big(p^1 A_d(1)\otimes A_d(1)\big)<1$. Since $A^{\diamond}(0) = A^{\diamond}(1) = A, Q^{\diamond}(0) = Q^{\diamond}(1) = Q$ and $(A,Q)$ is detectable from Assumption \ref{assum:det_stb_2C},  there exist matrices $H^{\diamond}(0)$ and $H^{\diamond}(1)$ such that $\rho\big(A_d(0)\otimes A_d(0)\big)<1$ and $\rho\big(p^1 A_d(1)\otimes A_d(1)\big)<1$. Hence, the MJLS is SD.

Thus, the MJLS of \eqref{eq:MJLS_A}-\eqref{eq:MJLS_theta} is SD for any $p^1$ and it is SS if and only if $p^1 < p_c^1$. It then follows from Lemma \ref{lm:MJ_infinite} that matrices $P^{\diamond}_{t}(m)$, $m \in \{0,1\}$, converge as $t \to -\infty$ to
 PSD matrices  $P^{\diamond}_{*}(m), m \in \{0,1\}$ that satisfy the steady state version of \eqref{P_MJ_cmp_2C_0}-\eqref{P_MJ_cmp_2C_1} (i.e, equations \eqref{eq:P_finite_2C_fixed} - \eqref{eq:tildeP_finite_2C_fixed}) if and only if $p^1 < p_c^1$.  This proves the lemma.

\section{Proof of Lemma \ref{lm:Q2_2C}, parts 1 and 2}\label{sec:lm_Q2_2C}
Let $g^*$ denote the strategies described  by \eqref{eq:K_finite_2C_fixed}-\eqref{eq:estimator_inf_2C}. We want to show that for any $g \in \mathcal G$, $J_{\infty}(g) \geq J_{\infty}(g^*)$ and that  $J_{\infty}(g^*)$ is finite. We will make use of the following claim. 
\begin{claim}\label{claim:cost_optimal_2C}
For the strategies $g^*$ described  by \eqref{eq:K_finite_2C_fixed}-\eqref{eq:estimator_inf_2C}, the following equation is true:
\begin{align}
J_{T}(g^{*}) =&(T+1) \tr (\Lambda_*) - \ee^{g*} [V_{T+1} ],
\label{eq:costforgstar_2C}
\end{align}
where $J_T(g^*)$ is the finite horizon cost of $g^*$ over a horizon of duration $T$, $\Lambda_* = (1-p^1) P_{*}^0+p^1 P_{*}^1$ and for any $t\geq 0$,
\begin{align}
V_t = \hat X_{t}^\tp P_*^0 \hat X_t + \tr \big(P_*^{1}\cov(X_{t}|H^0_{t}) \big).
\label{V_t_2C}
\end{align}
\end{claim}
\begin{proof}
See Appendix \ref{proof_claim:cost_optimal_2C} for a proof of this claim.
\end{proof}
Based on Claim \ref{claim:cost_optimal_2C}, the infinite horizon average cost  for $g^*$ is given as 
\begin{align}
J_{\infty}(g^*)
= &\limsup_{T\rightarrow\infty} \frac{1}{T+1} J_{T}(g^*)
\notag\\
= & \tr (\Lambda_*)  - \liminf_{T\rightarrow\infty} \frac{\ee^{g^*}[V_{T+1}]}{T+1}
\leq \tr (\Lambda_*),
\label{bound_J_infty}
\end{align}
where the last inequality holds because $V_{T+1} \geq 0$. 

For $n =0,1$, define $Y_{0}^n = 0$, and for $k=0,1,2,\dots,$
\begin{align}
&Y_{k+1}^0 =  \Omega(Y_{k}^0,Q,R,A,B),
\label{eq:Y_k_0_2C}
\\
&Y_{k+1}^1 =  \Omega \big((1-p^1) Y_{k}^0+p^1 Y_{k}^1,Q,R^{11},A,B^{11}\big).
\label{eq:Y_k_n_2C}
\end{align}
It's easy to check that for  $n=0,1,$  $Y_k^n = P_{T+1-k}^n$ for all  $k \geq 0$, and that $\lim_{k \rightarrow \infty} Y_k^n = \lim_{t \rightarrow -\infty} P_t^n  = P^n_*$.

Further, let's define  $\Lambda_k =  (1-p^1) Y_{k}^0+p^1 Y_{k}^1$. From \eqref{eq:opcost_finite_2C} of Lemma \ref{lm:opt_strategies_2C}, we know that the optimal finite horizon cost is given as 
\begin{align}
J^*_{T} &=  \sum_{t = 0}^T \tr \big((1-p^1)P_{t+1}^0+p^1 P_{t+1}^1 \big) \notag \\
&=\sum_{k = 0}^T \tr \big((1-p^1)P_{T+1-k}^0+p^1 P_{T+1-k}^1 \big) \notag \\
&=\sum_{k = 0}^T \tr \big((1-p^1)Y_k^0+p^1 Y_k^1 \big) \notag \\
&= \sum_{k = 0}^T \tr(\Lambda_k). \label{eq:Jstar}
\end{align}


We can therefore write
\begin{align}
 \lim_{T\rightarrow \infty}\frac{1}{T+1} J^*_{T} 
  &=  \lim_{T\rightarrow \infty}\frac{1}{T+1} \sum_{k = 0}^T \tr (\Lambda_k)
= \tr (\Lambda^*),
\label{limit_optimal_finite_cost}
\end{align}
where the last equality is correct because $\lim_{k \rightarrow \infty} Y_k^n = P^n_*$ for $n=0,1$.

Now, for any $g \in \mathcal G$,
\begin{align}
J_{\infty}(g)
= &\limsup_{T\rightarrow\infty} \frac{1}{T+1} J_{T}(g)
\notag\\
\geq &\limsup_{T\rightarrow\infty} \frac{1}{T+1} J_{T}^*
= \tr (\Lambda^*)
\geq J_{\infty}(g^*),
\label{eq:lowerbound}
\end{align}
where the first inequality is true because by definition $J_T^* = \inf_{g' \in\mathcal{G}} 
J_T(g') \leq J_T(g)$ for any $g \in \mathcal G$,  the  second equality is true because of \eqref{limit_optimal_finite_cost} and the last inequality is true because of \eqref{bound_J_infty}.
Hence, $g^*$ is optimal for Problem \ref{problem_infinite_2C}, and the optimal cost is finite and equal to $\tr (\Lambda^*)$.

%


\section{Proof of Claim \ref{claim:cost_optimal_2C}}
\label{proof_claim:cost_optimal_2C}
In order to show that \eqref{eq:costforgstar_2C} holds, it suffices to show that the following equation is true for all $t \geq 0$:
\begin{align}
\ee^{g^*} [ c(X_t,U_t^*) | H_t^0]=  \tr(\Lambda_*)  + \ee^{g^*} [V_t - V_{t+1} | H_t^0],
\label{main_goal_2C}
\end{align}
where $U^*_t$ are the control actions at time $t$ under $g^*$.
This is because by taking the expectation of \eqref{main_goal_2C} and summing it from $t=0$ to $T$, we obtain
\begin{align}
J_{T}(g^*) 
= & \sum_{t=0}^T \ee^{g^*} [ c(X_t,U_t^*) ]
\notag\\
=&(T+1) \tr(\Lambda_*) 
+ \ee^{g^*} [V_0 - V_{T+1} ]
\notag\\
=&(T+1) \tr(\Lambda_*) 
- \ee^{g^*} [V_{T+1} ],
\end{align}
where the last equality holds because $V_0 = 0$ (recall that $X_0 = \hat X_0 =0$). 


Now, to show that \eqref{main_goal_2C} holds, first note that $\ee^{g^*}[V_t | H_t^0] = V_t$ since $V_t$, given by \eqref{V_t_2C}, is a function of $H^0_t$.  Hence, \eqref{main_goal_2C} is equivalent to
\begin{align}
\ee^{g^*} [ c(X_t,U_t^*) | H_t^0] + \ee^{g^*} [ V_{t+1} | H_t^0] = \tr(\Lambda_*)  + V_t. \label{eq:main_goal_2C_2}
\end{align}
 In the following subsections we will calculate $\ee^{g^*} [ c(X_t,U_t^*) | H_t^0]$ and $\ee^{g^*}[V_{t+1} | H_t^0]$ and then simplify the left hand side of \eqref{eq:main_goal_2C_2}.  To do so, we define 
\begin{align}
&\hat X_{t+1|t} := \ee [X_{t+1}|H^0_{t}]
\\
&\Sigma_{t+1|t} := \cov(X_{t+1}|H^0_{t})
 \\
&\Sigma_t := \cov(X_{t}|H^0_{t}),
\label{Sigma_2C}
\end{align}
and recall that $\hat{X}_t = \ee [X_{t}|H^0_{t}]$.

\subsection{Calculating $\ee^{g^*} [ c(X_t,U_t^*) | H_t^0]$}
Note that 
\begin{small}
\begin{align}
&\ee^{g^*} [ c(X_t,U_t^*) | H_t^0]=  \underbrace{\ee^{g^*} [X_t^\tp Q X_t | H_t^0]}_{\mathbb{T}_4} + \underbrace{\ee^{g^*}[U_t^{*\tp} R U_t^* | H_t^0]}_{\mathbb{T}_5}. 
\label{cost_given_h_2C}
\end{align}
\end{small}
We can simplify the term $\mathbb{T}_4$ as follows
\begin{small}
\begin{align}
\mathbb{T}_4 &=  \hat X_t^\tp Q\hat X_t + \ee^{g^*} \Big[(X_t-\hat X_t)^\tp Q (X_t-\hat X_t)|H^0_t\Big] \notag \\
& =  \hat X_t^\tp Q\hat X_t + \tr\big(Q\Sigma_t \big). \label{eq:T4eq}
\end{align}
\end{small}

 From \eqref{eq:opt_U_2C_fixed}, we have $U_t^* = K_*^0\hat X_t + \tilde K_*^1 (X_t - \hat X_t)$, where $\tilde K_*^1 =  \bmat{\mathbf{0} \\ K_*^1 }$.  Therefore, we can simplify the term $\mathbb{T}_5$ as follows
\begin{small}
\begin{align}
\mathbb{T}_5 &= (K_*^0\hat X_t)^\tp R K_*^0\hat X_t \notag \\
& + \ee^{g^*} \Big[(X_t-\hat X_t)^\tp (\tilde K_*^1)^{\tp} R \tilde K_*^1 (X_t-\hat X_t)|H^0_t\Big]
\notag\\
&= \hat X_t^\tp(K_*^0)^\tp R K_*^0\hat X_t + \tr \big( (\tilde K_*^1)^\tp R \tilde K_*^1 \Sigma_t \big).
 \label{cost_on_U_t_2C}
\end{align}
\end{small}

Putting \eqref{cost_given_h_2C}, \eqref{eq:T4eq} and \eqref{cost_on_U_t_2C} together, we can write 
\begin{align}\label{eq:calculated_c}
&\ee^{g^*} [ c(X_t,U_t^*) | H_t^0]= \hat X_t^\tp Q\hat X_t +\tr(Q \Sigma_t) \notag \\
&+ \hat X_t^\tp(K_*^0)^\tp R K_*^0\hat X_t + \tr( (\tilde K_*^1)^\tp R \tilde K_*^1 \Sigma_t).
\end{align}

\subsection{Calculating $\ee^{g^*}[V_{t+1}|H^0_t]$}
From the definition of $V_{t+1}$ (see \eqref{V_t_2C})  we have
\begin{small}
\begin{align}
 \ee^{g^*}[V_{t+1}|H^0_t] &= \underbrace{\ee^{g^*}\Big[
\hat X_{t+1}^\tp P_*^0 \hat X_{t+1} \Big| H^0_t
\Big]}_{\mathbb{T}_6} \notag \\
&+ \underbrace{\ee^{g^*}\Big[
\tr \big(P_*^{1} \Sigma_{t+1} \big) \Big| H^0_t \Big]}_{\mathbb{T}_7}.\label{eq:T6T7}
\end{align}
\end{small}
Note that if $\Gamma_{t+1}^1 = 1$ (i.e., the link is active) $\hat X_{t+1} = X_{t+1}$ and   $\Sigma_{t+1} = 0$  and if $\Gamma_{t+1}^1 = 0$ $\hat X_{t+1} = \hat X_{t+1|t}$ and $\Sigma_{t+1} = \Sigma_{t+1|t}$\footnote{If $\Gamma_{t+1}^1=0$, the remote controller gets no new information about $X_{t+1}$. Hence, its belief on $X_{t+1}$ given $H_{t+1}^0$ remains the same as its belief on $X_{t+1}$ given $H_{t}^0$.}. That is, 
\begin{align}
\label{estimation_state_2C}
&\hat X_{t+1} = \Gamma_{t+1}^1 X_{t+1} +(1-\Gamma_{t+1}^1) \hat X_{t+1|t},  \\
 &\Sigma_{t+1}  = (1-\Gamma_{t+1}^1) \Sigma_{t+1|t}.
\label{estimation_covariance_2C}
\end{align}
Now, we use \eqref{estimation_state_2C} and \eqref{estimation_covariance_2C} to calculate  the terms $\mathbb{T}_6$ and $\mathbb{T}_7$ in \eqref{eq:T6T7}.

Note that from \eqref{estimation_state_2C}, $\hat X_{t+1}$ is equal to $\hat X_{t+1|t} + \Gamma_{t+1}(X_{t+1} -\hat X_{t+1|t})$. Therefore,  $\mathbb{T}_6$ can be written as

\vspace{-2mm}
\begin{small}
\begin{align}
&\mathbb{T}_6= 
\ee^{g^*}\Big[
\hat X_{t+1|t} ^\tp P_*^0 \hat X_{t+1|t} \Big| H^0_t \Big] \notag \\
&+ \ee^{g^*}\Big[
(X_{t+1} -\hat X_{t+1|t})^\tp \Gamma_{t+1}^1 P_*^0 \Gamma_{t+1}^1 (X_{t+1} -\hat X_{t+1|t}) \Big| H^0_t \Big] \notag \\ & = 
\hat X_{t+1|t} ^\tp P_*^0 \hat X_{t+1|t}+ (1-p^1) \tr (P_*^0 \Sigma_{t+1|t}).
\label{eq:T_4_1}
\end{align}
\end{small}
Furthermore, using \eqref{estimation_covariance_2C}, it is straightforward to see that 

\vspace{-2mm}
\begin{small}
\begin{align}
&\mathbb{T}_7
= p^1 \tr (P_*^1  \Sigma_{t+1|t}).\label{eq:T7}
\end{align}
\end{small}

Combining \eqref{eq:T6T7}, \eqref{eq:T_4_1} and \eqref{eq:T7}, we get
\begin{align}
 \ee^{g^*}[V_{t+1}|H^0_t]= &\hat X_{t+1|t} ^\tp P_*^0 \hat X_{t+1|t}+ (1-p^1) \tr (P_*^0 \Sigma_{t+1|t})\notag \\
& +  p^1 \tr (P_*^1  \Sigma_{t+1|t}). \label{eq:vt1}
\end{align}
Note that the right hand side of  \eqref{eq:vt1} involves  $\hat X_{t+1|t}$ and $\Sigma_{t+1|t}$. We will now try to write these in terms of  $\hat X_{t}$ and $\Sigma_{t}$. For that purpose,  note that under the strategies $g^*$
\vspace{-2mm}
\begin{small}
\begin{align}
X_{t+1} &= AX_t + B \Big[ 
K_*^0\hat X_t + \bmat{\mathbf{0} \\ K_*^1 }(X_t - \hat X_t)
 \Big] + W_t \notag \\
& = A_s(0) \hat X_t + A_s(1)  ( X_t - \hat X_t) + W_t.
\label{X_dynamics_proof_2C}
\end{align}
\end{small}
In the above equation, we have defined $A_s(0) = A+ B K_*^0$ and $A_s(1) = A + B \tilde K_*^1$ where   $\tilde K_*^1 =  \bmat{\mathbf{0} \\ K_*^1 }$. Now using \eqref{X_dynamics_proof_2C}, we can calculate $\hat X_{t+1|t}$ and $\Sigma_{t+1|t}$ as follows, 
\begin{align}
\label{estimation_X_2C}
&\hat X_{t+1|t} = A_s(0) \hat X_t, \\
&\Sigma_{t+1|t} =  \mathbf I + A_s(1) \Sigma_t A_s(1)^\tp.
\label{Sigma_t_2C}
\end{align}

Using \eqref{estimation_X_2C} and \eqref{Sigma_t_2C} in \eqref{eq:vt1}, we get
\begin{align}
 \ee^{g^*}[V_{t+1}|H^0_t]&= \hat X_t ^\tp A_s(0)^\tp P_*^0 A_s(0) \hat X_t+ (1-p^1) \tr (P_*^0) \notag \\
& +  p^1 \tr (P_*^1) + (1-p^1) \tr(P_*^0A_s(1) \Sigma_t A_s(1)^\tp) \notag \\
&+ p^1\tr(P_*^1A_s(1) \Sigma_t A_s(1)^\tp). \label{eq:vt1A}
\end{align}
Recall that $\Lambda_* = (1-p^1) P_{*}^0+p^1 P_{*}^1$. Thus, \eqref{eq:vt1A} can be written as
\begin{align}
 \ee^{g^*}[V_{t+1}|H^0_t]&= \hat X_t ^\tp A_s(0)^\tp P_*^0 A_s(0) \hat X_t+  \tr (\Lambda_*) \notag \\
& +   (1-p^1) \tr(P_*^0A_s(1) \Sigma_t A_s(1)^\tp) \notag \\
&+ p^1\tr(P_*^1A_s(1) \Sigma_t A_s(1)^\tp). \label{eq:vt1B}
\end{align}

\subsection{Simplifying the left hand side of \eqref{eq:main_goal_2C_2}}
Now that we have calculated $\ee^{g^*} [ c(X_t,U_t^*) | H_t^0]$ and $\ee^{g^*}[V_{t+1} | H_t^0]$, we will try to simplify the left hand side of \eqref{eq:main_goal_2C_2}.

Adding \eqref{eq:vt1B} and \eqref{eq:calculated_c} and grouping together the terms involving $\hat X_t$ and those involving $\Sigma_t$, we can write
\begin{align}
& \ee^{g^*}[c(X_t,U_t^*)|H^0_t] + \ee^{g^*}[V_{t+1}|H^0_t] 
= \tr(\Lambda_*) \notag \\
& + \hat X_t ^\tp \Phi(P_*^0,K^0_*) \hat X_t  + \tr(\Phi((1-p^1)P_*^0 +p^1P^1_*,\tilde K^1_*)\Sigma_t), \label{eq:vt1C}
\end{align}
where 
\begin{small}
\begin{align}
\Phi(P_*^0,K^0_*) := Q + (K^0_*)^\tp RK^0_*
+ (A+BK^0_*)^\tp P^0_*(A+BK^0_*),
\end{align}
\end{small}
and similarly
\begin{small}
\begin{align}
&\Phi((1-p^1)P_*^0 +p^1P^1_*,\tilde K^1_*):=Q + (\tilde K_*^1)^\tp R \tilde K_*^1 \notag \\
&~~+(A+B\tilde K^1_*)^\tp\big((1-p^1)P_*^0 +p^1P^1_* \big)(A+B\tilde K^1_*) \notag \\
&= Q + ( K_*^1)^\tp R^{11}  K_*^1 \notag \\
&~~+(A+B^{11} K^1_*)^\tp\big((1-p^1)P_*^0 +p^1P^1_* \big)(A+B^{11}K^1_*) 
\end{align}
\end{small}

Using the fact that $K^0_* = \Psi(P_{*}^0,R,A,B)$, it can be easily established that 
\begin{equation}
\Phi(P_*^0,K^0_*) = \Omega(P_{*}^0,Q,R,A,B).
\end{equation}
Further, since $P^0_* = \Omega(P_{*}^0,Q,R,A,B)$, we have 
\begin{equation}
\Phi(P_*^0,K^0_*) = P^0_*.\label{eq:phi0}
\end{equation}

Similarly, using the fact that $K_*^1 = \Psi((1-p^1)P_{*}^0+p^1 P_{*}^1,R^{11},A,B^{11})$, it can be established that 
\begin{align}
&\Phi((1-p^1)P_*^0 +p^1P^1_*,\tilde K^1_*) \notag \\
&= \Omega((1-p^1)P_*^0 +p^1P^1_*,Q,R^{11},A,B^{11}).
\end{align}
Further, since $P^1_* = \Omega((1-p^1)P_*^0 +p^1P^1_*,Q,R^{11},A,B^{11})$, we have 
\begin{equation}
\Phi((1-p^1)P_*^0 +p^1P^1_*,\tilde K^1_*) =  P^1_*.\label{eq:phi1}
\end{equation}

Using \eqref{eq:phi0} and \eqref{eq:phi1} in \eqref{eq:vt1C}, we get
\begin{align}
& \ee^{g^*}[c(X_t,U_t^*)|H^0_t] + \ee^{g^*}[V_{t+1}|H^0_t]  \notag \\
&= \tr(\Lambda_*)
 + \hat X_t ^\tp P^0_* \hat X_t  + \tr(P^1_*\Sigma_t), \notag \\
&= \tr(\Lambda_*) + V_t. \label{eq:vt1D}
\end{align}
This establishes \eqref{eq:main_goal_2C_2} and hence completes the proof of the claim.

\section{Proof of Lemma \ref{lm:Q2_2C}, part 3}\label{sec:stability_proof}
Let $\tilde X_t := X_t - \hat X_t$ denote the estimation error.  It suffices to show that $\hat X_t$ and  $\tilde X_t$ are mean square stable.
The optimal strategies can be written as 
\begin{align}
U_t = K^0_* \hat X_t + 
\begin{bmatrix}
0 \\ K^1_* 
\end{bmatrix}
\tilde X_t.
\end{align}
Then, from \eqref{eq:estimator_inf_2C} we have
\begin{align}
\tilde X_{t+1} 
 = & (1-\Gamma_{t+1}^1)(A_s(1) \tilde X_t + W_t),
\end{align}
where $A_s(1) = (A + B^{11}K^1_*)$.
If $p^1 = 0$ or $1$, the stability result follows from standard  linear system theory arguments.
If $0< p^1 < 1$, the estimation error $\tilde X_{t}$ is a MJLS with an i.i.d. switching process\footnote{Note that this MJLS is not the same as the auxiliary MJLS constructed in Section \ref{sec:Q1_2C}.}.
From \cite[Theorem 3.33]{costa2006discrete}, the estimation error process is mean square stable if the corresponding noiseless system (i.e., with $W_t = 0$) is mean square stable.
Because $\Gamma_{t+1}^1$ is an i.i.d. process, from \cite[Corollary 2.7]{fang2002stochastic}, the noiseless system is mean-square stable if
$p^1\rho( A_{s}(1) \otimes A_{s}(1)) < 1$.

Note that the gain matrices $K^0_* ,K^1_* $ are obtained from the DCARE in \eqref{eq:CARE_infinite} for the SD and SS auxiliary MJLS described by \eqref{eq:MJLS_A}-\eqref{eq:MJLS_theta}, so the corresponding gains stabilize the auxiliary MJLS \cite[Corollary A.16]{costa2006discrete}, \cite[Theorem A.15]{costa2006discrete}. That is, the following matrix
\begin{align}
\mathcal{A}_s = \begin{bmatrix}
A_s(0)\otimes A_s(0) & (1-p^1)A_s(1)\otimes A_s(1) \\
 \mathbf{0} & p^1A_s(1)\otimes A_s(1)
\end{bmatrix}
\label{eq:A_smat}
\end{align}
has a spectral radius less than one (see the proof of Lemma \ref{lm:pc_2C}),  where $A_s(0) = A + BK^0_*$. Thus, $\rho(A_s(0)) <  1$ and $p^1\rho( A_{s}(1) \otimes A_{s}(1)) < 1$. Consequently, the estimation error $\tilde X_{t}$ is mean-square stable.

Now, note that the estimate evolution can be written as
\begin{align}
\hat X_{t+1} 
 = & A_s(0) \hat X_t + \tilde W_t,
\end{align}
where $\tilde W_t = \Gamma_{t+1}^1(A_s(0) \tilde X_t + W_t)$ can be viewed as a ``noise" process.
The process $\tilde W_t$ is mean square stable because $\Gamma_{t+1}^1 \leq 1$, and $\tilde X_t$ and $W_t$ are both mean square stable.
Since $\rho(A_s(0)) <  1$, we conclude that $\hat X_t$ is mean square stable using standard linear system arguments \cite[Theorem 3.4]{KumarVaraiya:1986}.

%

\section{Proof of Lemma \ref{lm:Q3}} \label{sec:lmQ3}
 Consider the matrices $Y^n_k$, $n=0,1, k =0,1,\ldots,$ defined by $Y^n_0 =0$ and the recursions in \eqref{eq:Y_k_0_2C} and \eqref{eq:Y_k_n_2C}. Since  matrices $P_t^n$, $n =0,1,$ do not converge as $t \to -\infty$, it follows that  matrices $Y_k^n$, $n=0,1$, do not converge as $k \rightarrow \infty$ (recall that $Y_k^n = P_{T+1-k}^n $ for $n=0,1$ and $k \geq 0$). 
  
 Recall from  \eqref{eq:Jstar} that $J^*_T = \sum_{k = 0}^T \tr(\Lambda_k)$, where $\Lambda_k =  (1-p^1) Y_{k}^0+p^1 Y_{k}^1$. Also, from the first inequality in \eqref{eq:lowerbound}, recall that $J_{\infty}(g) \geq \limsup_{T\rightarrow\infty} \frac{1}{T+1} J_{T}^*$ for any strategy $g$. Therefore, to show that no strategy can achieve finite cost, it suffices to show that 
 \begin{equation}
 \limsup_{T\rightarrow\infty} \frac{1}{T+1} J_{T}^* =  \limsup_{T\rightarrow\infty} \frac{1}{T+1} \sum_{k = 0}^T \tr(\Lambda_k) = \infty.
 \end{equation}

To do so, we first show that the sequence $\{Y_k^n, k=0,1,\ldots\}$ is monotonically increasing\footnote{ in the sense of the partial order $\preceq$.} for $n =0,1$. To this end, note that $Y_1^n \succeq Y_0^n=0$ for $n \in \{0,1\}$. Furthermore, the monotonic property of the operator $\Omega(\cdot)$ (proved in part (ii) of Lemma \ref{lm:KPrelation_0} in Appendix \ref{app:lm:costgstar}) implies that for $n=0,1,$ and for all $k \geq 0$, $Y_{k+1}^n \succeq Y_{k}^n$. Now, if the sequences $\{Y_k^n, k\geq 0\}$,  $n =0,1,$ are bounded, they will converge due to the monotone behavior. This contradicts the fact that these sequences do not converge as $k \to \infty$. Therefore, at least one of the two sequences $\{Y_k^0, k\geq0\}$, $\{Y_k^1, k\geq0\}$  is unbounded. Consequently, the sequence $\{\Lambda_k, k\geq 0\}$ is unbounded. Hence, $\limsup_{T\rightarrow\infty} \frac{1}{T+1} \sum_{k = 0}^T \tr(\Lambda_k) = \infty.$ This completes the proof.

\section{Proof of Lemma \ref{lm:opt_strategies_new_rep}}
\label{proof_lm:opt_strategies_new_rep}

By comparing \eqref{eq:P_N_init}-\eqref{eq:P_finite} with  \eqref{eq:barP_init}-\eqref{eq:barP_finite_0}, it is straightforward to observe that $P_t^0 = \bar P_t^0$ for all $t$. We will now show by induction that at any time $t$, $P_t^n = \bar P_t^n$ for $n =1,\ldots,N$.
First note that by definition, $P_{T+1}^{n} = \mathbf{0}$ and $\bar P^n_{T+1} = \mathbf{0}$ for $n =1,\ldots,N$. Hence, \eqref{new_rep2} is correct at time $T+1$. Now, assume that \eqref{new_rep2} is correct at time $t+1$ (induction hypothesis). Then, from \eqref{eq:barP_finite} and  the induction hypothesis, we have for $n =1,\ldots,N$,
\begin{align}
\bar P_t^{n} &=  \Omega \big((1-p^n) P_{t+1}^0+p^n \mathcal{L}_{zero}(P^0_{t+1}, P_{t+1}^{n},n,n), \notag \\
& \hspace{1.0cm}  \mathcal{L}_{zero}(Q,Q^{nn},n,n), \mathcal{L}_{iden}(R,R^{nn},n+1), \notag \\
 & \hspace{1.0cm}  \mathcal{L}_{zero}(A,A^{nn},n,n), \mathcal{L}_{zero}(B,B^{nn},n,n+1) \big) \notag \\
&=
\mathcal{L}_{zero}(Q,Q^{nn},n,n)  + \mathbb{T}_1 - \mathbb{T}_2 (\mathbb{T}_3)^{-1} (\mathbb{T}_2)^{\tp},
\label{P_MJ_n}
\end{align}
where 
\begin{align}
\label{T_1_def}
&\mathbb{T}_1 = \mathcal{L}_{zero}(A,A^{nn},n,n)^{\tp} \bar{\bar{P}}_{t+1} \mathcal{L}_{zero}(A,A^{nn},n,n) \\
\label{T_2_def}
&\mathbb{T}_2 = \mathcal{L}_{zero}(A,A^{nn},n,n)^{\tp} \bar{\bar{P}}_{t+1} \mathcal{L}_{zero}(B,B^{nn},n,n+1) , \\
\label{T_3_def}
&\mathbb{T}_3 = \mathcal{L}_{iden}(R,R^{nn},n+1) 
\notag \\
&+ \mathcal{L}_{zero}(B,B^{nn},n,n+1)^{\tp} \bar{\bar{P}}_{t+1} \mathcal{L}_{zero}(B,B^{nn},n,n+1),
\end{align}
and we have defined $\bar{\bar{P}}_{t+1} = (1-p^n) P_{t+1}^0+p^n \mathcal{L}_{zero}(P^0_{t+1}, P_{t+1}^{n},n,n)$.

Note that from the definitions of operators $\mathcal{L}_{zero}$ and $\mathcal{L}_{iden}$ in \eqref{L_zero}-\eqref{L_iden}, it is straightforward to observe that the block dimensions of $\mathbb{T}_1, \mathbb{T}_2, \mathbb{T}_3$ are the same as the block dimensions of $A,B,B^{\tp} B$, respectively (They are block matrices of sizes $N \times N$, $N \times (N+1)$, and $(N+1) \times (N+1)$, respectively). Therefore, through straightforward algebraic manipulations, we can get
\begin{align}
\label{T_1}
&\mathbb{T}_1 = \mathcal{L}_{zero}(A,\mathbb{\tilde T}_1,n,n), \\
\label{T_2}
&\mathbb{T}_2 =  \mathcal{L}_{zero}(B,\mathbb{\tilde T}_2,n,n+1), \\
\label{T_3}
&\mathbb{T}_3 = \mathcal{L}_{iden}(B^{\tp}B,\mathbb{\tilde T}_3,n+1),
\end{align}
where
\begin{align}
\label{tilde_T_1}
&\mathbb{\tilde T}_1 = (A^{nn})^{\tp}[(1-p^n)[P_{t+1}^0]_{n,n}+p^n P_{t+1}^{n}]A^{nn}, \\
\label{tilde_T_2}
&\mathbb{\tilde T}_2 = (A^{nn})^{\tp}[(1-p^n)[P_{t+1}^0]_{n,n}+p^n P_{t+1}^{n}]B^{nn}, \\
\label{tilde_T_3}
&\mathbb{\tilde T}_3 = R^{nn} + (B^{nn})^{\tp}[(1-p^n)[P_{t+1}^0]_{n,n}+p^n P_{t+1}^{n}]B^{nn}.
\end{align}

Further, since $\mathbb{T}_3$ is a block diagonal matrix, we have 
\begin{align}
& (\mathbb{T}_3)^{-1} =\mathcal{L}_{iden} \big(B^{\tp}B,(\mathbb{\tilde T}_3)^{-1},n+1 \big).
\label{T_3_inv}
\end{align}

Now, using \eqref{T_1}-\eqref{T_3_inv} and the fact that  matrices $A, Q, BB^{\tp}$ have the same size as matrix $P^0_t$ (They are block matrices of size $N \times N$), \eqref{P_MJ_n} can be simplified to
\begin{align}
\bar P_t^{n} &= \mathcal{L}_{zero} \big(P^0_t,  Q^{nn} + \mathbb{\tilde T}_1 - \mathbb{\tilde T}_2 (\mathbb{\tilde T}_3)^{-1} (\mathbb{\tilde T}_2)^{\tp},n,n \big) \notag \\
&=   \mathcal{L}_{zero}(P^0_t, P_t^{n},n,n),
\end{align}
where the last equality is true because of the definition of $P_t^{n}$ in \eqref{eq:tildeP_finite}. Hence, \eqref{new_rep2} is true at time $t$. This completes the proof.

\section{Proof of Lemma \ref{lm:pc_NC}}
\label{proof_lm:pc}

\begin{figure*}[t]
\begin{small}
\begin{align}
\mathcal{A}_s =
\begin{blockarray}{ccccc}
\begin{block}{[ccccc]}
                     A_s(0)\otimes A_s(0) & (1-p^1)A_s(1)\otimes A_s(1)& (1-p^2)A_s(2)\otimes A_s(2) & \ldots & (1-p^n)A_s(N)\otimes A_s(N) \\ \\
                     \mathbf{0} & p^1A_s(1)\otimes A_s(1) &\mathbf{0} &\ldots & \mathbf{0} \\ \\
                     \vdots &\ddots  &p^2 A_s(2)\otimes A_s(2)  & \ddots & \vdots \\ \\
                     \vdots &   &\ddots  &\ddots & \mathbf{0} \\ \\
                     \mathbf{0} & \ldots &\ldots & \mathbf{0}  & p^n A_s(N)\otimes A_s(N)\\
\end{block}
\end{blockarray}
\label{eq:bigmatrix3}
\end{align}
\end{small}
\hrule
\end{figure*}
%
\begin{figure*}[t]
\begin{small}
\begin{align}
A_s(n) &= A^{\diamond}(n) + B^{\diamond}(n) K^{\diamond}(n)  \notag \\
&= \begin{blockarray}{cccl}
\text{$1:n-1$} &n &\text{$n+1:N$}  &  \\
\begin{block}{[ccc]l}
                      \text{\large 0}   & \text{\large 0} &\text{\large 0} & \text{$1:n-1$} \\
   B^{nn} [K^{\diamond}(n)]_{n+1,1:n-1} & A^{nn} + B^{nn} [K^{\diamond}(n)]_{n+1,n} & B^{nn} [K^{\diamond}(n)]_{n+1,n+1:N} &n \\
                      \text{\large 0}   & \text{\large 0} & \text{\large 0}&\text{$n+1:N$}  \\
\end{block}
\end{blockarray}.
\label{A_c_l}
\end{align}
\end{small}
\hrule
\end{figure*} 

From Lemma \ref{lm:opt_strategies_new_rep}, we know that the convergence of matrices $P_t^{n}$, $n \in \mathcal{\overline N}$, is equivalent to the convergence of matrices $\bar P_t^{n}$, $n \in \mathcal{\overline N}$. Further, because of Lemma \ref{equality_recursions_NC}, $\bar P^n_t = P^{\diamond}_t(n)$, $n \in \mathcal{\overline N}$, where matrices $P^{\diamond}_t(n)$, $n \in \mathcal{\overline N}$, are defined by \eqref{eq:P_N_MJ_init}-\eqref{P_MJ_cmp_NC_1} for the auxiliary MJLS. Thus, in order to study the the convergence of matrices $P_t^{n}$, $n \in \mathcal{\overline N}$, we can focus on the convergence of matrices $P^{\diamond}_t(n)$, $n \in \mathcal{\overline N}$.

To investigate the convergence of $P^{\diamond}_t(n)$, $n \in \mathcal{\overline N}$, we first show that the auxiliary MJLS described by \eqref{A_mj}-\eqref{transition_prob} is SS if and only if $p^n < p_c^n$ for all $n \in \mathcal{N}$. To do so, we can follow a methodology similar to the one used to prove Lemma \ref{lm:pc_2C} with $\mathcal{A}_s$ defined as in \eqref{eq:bigmatrix3} and $A_s(n)$, $n =1,\ldots,N,$ defined as in \eqref{A_c_l}. 

Next, we can use part 2 of Lemma \ref{lm:ss} to show that the auxiliary MJLS described by \eqref{A_mj}-\eqref{transition_prob} is SD if and only of $\mathcal{A}_d$ defined as in \eqref{eq:bigmatrix4} is Schur stable. Since the matrix $\mathcal{A}_{d}$ is upper-triangular, it is Schur stable if and only if there exist matrices $H^{\diamond}(0)$ and $H^{\diamond}(n)$ for $n \in \mathcal{N}$ such that $\rho\big(A_d(0)\otimes A_d(0)\big)<1$ and $\rho\big(p^n A_d(n)\otimes A_d(n)\big)<1$, where $A_d(n)$, $n \in \mathcal{N}$, defined as in \eqref{A_c_l_2}. The existence of these matrices follows from detectability of $(A,Q)$ and $\big(A^{nn},(Q^{nn})^{1/2} \big)$ for  $n \in \mathcal{N}$ (see Assumptions \ref{assum:det_stb} and \ref{assum:det_stb_2}). 
Hence, the MJLS is SD.


It then follows from Lemma \ref{lm:MJ_infinite} that matrices $P^{\diamond}_{t}(n)$, $n \in \overline{\mathcal{N}}$, converge as $t \to -\infty$  if and only if $p^n < p_c^n$ for all $n \in \mathcal{N}$.  Consequently,  matrices $ P^0_t,\ldots,  P^N_t$  converge as $t \to -\infty$ to matrices $P_*^0,\ldots, P_*^N$ that satisfy the coupled fixed point  equations \eqref{eq:P_finite_NC_fixed}-\eqref{eq:tildeP_finite_NC_fixed} if and only if  $p^n < p_c^n$ for all $n \in \mathcal{N}$. This proves the lemma.

\begin{figure*}[t]
\begin{small}
\begin{align}
\mathcal{A}_d =
\begin{blockarray}{ccccc}
\begin{block}{[ccccc]}
                     A_d(0)\otimes A_d(0) & (1-p^1)A_d(1)\otimes A_d(1)& (1-p^2)A_d(2)\otimes A_d(2) & \ldots & (1-p^n)A_d(N)\otimes A_d(N) \\ \\
                     \mathbf{0} & p^1A_d(1)\otimes A_d(1) &\mathbf{0} &\ldots & \mathbf{0} \\ \\
                     \vdots &\ddots  &p^2 A_d(2)\otimes A_d(2)  & \ddots & \vdots \\ \\
                     \vdots &   &\ddots  &\ddots & \mathbf{0} \\ \\
                     \mathbf{0} & \ldots &\ldots & \mathbf{0}  & p^n A_d(N)\otimes A_d(N)\\
\end{block}
\end{blockarray}
\label{eq:bigmatrix4}
\end{align}
\end{small}
\hrule
\end{figure*}
\begin{figure*}[t]
\begin{small}
\begin{align}
&A_d(n) = A^{\diamond} (n) +  H^{\diamond}(n) \big(Q^{\diamond}(n) \big)^{1/2}
= \begin{blockarray}{cccl}
\text{$1:n-1$} &n &\text{$n+1:N$}  &  \\
\begin{block}{[ccc]l}
                      \text{\large 0}   & [H^{\diamond}(n)]_{n,1:n-1} (Q^{nn})^{1/2} &\text{\large 0} & \text{$1:n-1$} \\
   \text{\large 0} & A^{nn} + [H^{\diamond}(n)]_{n,n} (Q^{nn})^{1/2} & \text{\large 0} &n \\
                      \text{\large 0}   & [H^{\diamond}(n)]_{n,n+1:N} (Q^{nn})^{1/2} & \text{\large 0}&\text{$n+1:N$}  \\
\end{block}
\end{blockarray}.
\label{A_c_l_2}
\end{align}
\end{small}
\hrule
\end{figure*} 


\section{Proof of Lemma \ref{lm:Q2_NC}} \label{sec:proof_lm_Q2_NC}
Let $g^*$ denote the strategies defined by \eqref{eq:K_finite_NC_fixed}-\eqref{eq:opt_U_NC_fixed}. We want to show that for any $g \in \mathcal G$, $J_{\infty}(g) \geq J_{\infty}(g^*)$ and that $J_{\infty}(g^*)$ is finite. We will make use of the following claim.

\begin{claim}\label{claim:cost_optimal_NC}
For the strategies $g^*$ described  by \eqref{eq:K_finite_NC_fixed}-\eqref{eq:opt_U_NC_fixed}, the following equation is true:
\begin{align}
J_{T}(g^{*}) =&(T+1) \tr (\Lambda_*) - \ee^{g*} [V_{T+1} ]
\label{eq:costforgstar_NC}
\end{align}
where $J_T(g^*)$ is the finite horizon cost of $g^*$ over a horizon of duration $T$, $\Lambda_* =  \sum_{n=1}^N \big( (1-p^n) [P_{*}^0]_{n,n}+p^n P_{*}^n \big)$ and for any $t\geq 0$,
\begin{align}
V_t = \hat X_{t}^\tp P_*^0 \hat X_t + \sum_{n=1}^N \tr \big(P_*^{n}\cov(X_{t}^n|H^0_{t}) \big).
\label{V_t_NC}
\end{align}
\end{claim}
\begin{proof}See Appendix \ref{Cost_of_the_Strategies} for a proof of Claim \ref{claim:cost_optimal_NC}.
\end{proof}

Along the lines of the proof of Lemma \ref{lm:Q2_2C} in Appendix \ref{sec:lm_Q2_2C}, we define  $Y_{0}^0 = 0$, and for $k=0,1,2,\dots,$
\begin{align}
&Y_{k+1}^0 =  \Omega(Y_{k}^0,Q,R,A,B),
\end{align}
and for $n=1,\ldots,N,$
\begin{align}
&Y^n_0 =0, \\
&Y_{k+1}^n =  \Omega \big((1-p^n)[Y_{k}^0]_{n,n}+p^n Y_{k}^{n},Q^{nn},R^{nn},A^{nn},B^{nn}\big).
\end{align}
It's easy to check that for  $n=0,1,\ldots,N,$  $Y_k^n = P_{T+1-k}^n$ for all  $k \geq 0$, and that $\lim_{k \rightarrow \infty} Y_k^n = \lim_{t \rightarrow -\infty} P_t^n  = P^n_*$.
 The rest of the proof for parts 1 and 2 follows the same arguments as in Appendix \ref{sec:lm_Q2_2C} for the proof of parts 1 and 2 of Lemma \ref{lm:Q2_2C}. 
For the proof of part 3, define for $n=1,\ldots,N,$ 
 \begin{align*}
 \bar K_*^n := 
 \Psi \big( &(1-p^n)\bar P_*^0+p^n \bar P_{*}^n, R^{\diamond}(n),  A^{\diamond}(n), B^{\diamond}(n) \big),
 \end{align*}
 where $\bar P_*^{0:N}$ are the limits of $\bar P_t^{0:N}$ (see Lemmas \ref{lm:opt_strategies_new_rep} and \ref{lm:pc_NC} and the auxiliary MJLS in Section \ref{sec:model_N_controllers}). 
Then, it can be shown that  (i) $\bar K_*^n = \mathcal{L}_{zero}(K_*^0,K_*^{n},n+1,n)$ and hence (ii)\\ $p^n\rho( (A^{nn} +B^{nn}K^n_*) \otimes (A^{nn} +B^{nn}K^n_*)) = p^n \rho (( A^{\diamond}(n) + B^{\diamond}(n)\bar K_*^n ) \otimes (A^{\diamond}(n) + B^{\diamond}(n)\bar K_*^n )),  $ which  is less than 1 since the auxiliary MJLS of \eqref{A_mj}-\eqref{transition_prob} is SD and SS (see proof of Lemma \ref{lm:pc_NC}). The rest of the proof uses similar arguments as in Appendix \ref{sec:stability_proof} for the proof of part 3 of Lemma 7.


\section{Proof of Lemma \ref{lm:Q3_NC}} \label{sec:proof_lm_Q3_NC}
The proof can be obtained by following the arguments  in the proof of Lemma \ref{lm:Q3} and defining $\Lambda_k =  \sum_{n=1}^N \big( (1-p^n) [Y_{k}^0]_{n,n}+p^n Y_{k}^n \big)$, where $Y^0_k,Y^n_k$ are as defined in Appendix \ref{sec:proof_lm_Q2_NC}.

\section{Proof of Claim \ref{claim:cost_optimal_NC}}
\label{Cost_of_the_Strategies}
In order to show that \eqref{eq:costforgstar_NC} holds, it suffices to show that the following equation is true for all $t \geq 0$:
\begin{align}
\ee^{g^*} [ c(X_t,U_t^*) | H_t^0]=  \tr(\Lambda_*)  + \ee^{g^*} [V_t - V_{t+1} | H_t^0],
\label{main_goal_NC}
\end{align}
where $U^*_t$ are the control actions at time $t$ under $g^*$.
This is because by taking the expectation of \eqref{main_goal_NC} and summing it from $t=0$ to $T$, we obtain
\begin{align}
J_{T}(g^*) 
=(T+1) \tr(\Lambda_*) 
- \ee^{g^*} [V_{T+1} ].
\end{align}

Now, to show that \eqref{main_goal_NC} holds,  note that $\ee^{g^*}[V_t | H_t^0] = V_t$ since $V_t$ is a function of $H^0_t$.  Hence, \eqref{main_goal_NC} is equivalent to
\begin{align}
\ee^{g^*} [ c(X_t,U_t^*) | H_t^0] + \ee^{g^*} [ V_{t+1} | H_t^0] = \tr(\Lambda_*)  + V_t. \label{eq:main_goal_NC_2}
\end{align}
 In the following subsections we will calculate $\ee^{g^*} [ c(X_t,U_t^*) | H_t^0]$ and $\ee^{g^*}[V_{t+1} | H_t^0]$ and then simplify the left hand side of \eqref{eq:main_goal_NC_2}.  To do so, we define for $n=1,\ldots,N,$
\begin{align}
&\hat X_{t+1|t}^n := \ee [X_{t+1}^n|H^0_{t}]
\\
&\Sigma_{t+1|t}^n := \cov(X_{t+1}^n|H^0_{t})
 \\
&\Sigma_t^n := \cov(X_{t}^n|H^0_{t}).
\label{sigma_t}
\end{align}
and recall that $\hat{X}_t^n = \ee [X_{t}^n|H^0_{t}]$.
 
 \subsection{Calculating $\ee^{g^*} [ c(X_t,U_t^*) | H_t^0]$}
 Note that 
\begin{small}
\begin{align}
&\ee^{g^*} [ c(X_t,U_t^*) | H_t^0]=  \underbrace{\ee^{g^*} [X_t^\tp Q X_t | H_t^0]}_{\mathbb{T}_4} + \underbrace{\ee^{g^*}[U_t^{*\tp} R U_t^* | H_t^0]}_{\mathbb{T}_5}. 
\label{cost_given_h_NC}
\end{align}
\end{small}
$\mathbb{T}_4$ can be written as
\begin{small}
\begin{align}
\mathbb{T}_4 &=  \hat X_t^\tp Q\hat X_t + \ee \Big[(X_t-\hat X_t)^\tp Q (X_t-\hat X_t)|H^0_t\Big] \notag \\
& =  \hat X_t^\tp Q\hat X_t + \sum_{n=1}^N \tr\big(Q^{nn}\Sigma_t^n \big).
\end{align}
\end{small}
where the second equality is true because according to \cite[Lemma 3]{asghari_ouyang_nayyar_tac_2018} $X_t^n$ and $X_t^m$, $m \neq n,$ are conditionally independent given $H_t^0$.

Similarly, $\mathbb{T}_5$ can be written as
\begin{small}
\begin{align}
\mathbb{T}_5 &=\hat X_t^\tp(K_*^0)^\tp R K_*^0\hat X_t + \sum_{n=1}^N \tr \big((K_*^n)^{\tp} R^{nn} K_*^n \Sigma_{t}^n \big ).
\end{align}
\end{small}
Thus,
\begin{align}\label{eq:calculated_c_N}
&\ee^{g^*} [ c(X_t,U_t^*) | H_t^0]= \hat X_t^\tp Q\hat X_t + \sum_{n=1}^N \tr\big(Q^{nn}\Sigma_t^n \big) \notag \\
&~+  \hat X_t^\tp(K_*^0)^\tp R K_*^0\hat X_t + \sum_{n=1}^N \tr \big((K_*^n)^{\tp} R^{nn} K_*^n \Sigma_{t}^n \big ).
\end{align} 
 
 \subsection{Calculating $\ee^{g^*}[V_{t+1}|H^0_t]$}
 From the definition of $V_{t+1}$ (see \eqref{V_t_NC})  we have

\vspace{-2mm}
\begin{small}
\begin{align}
 \ee^{g^*}[V_{t+1}|H^0_t] &= \underbrace{\ee^{g^*}\Big[
\hat X_{t+1}^\tp P_*^0 \hat X_{t+1} \Big| H^0_t
\Big]}_{\mathbb{T}_6} \notag \\
&+ \underbrace{\ee^{g^*}\Big[
\sum_{n=1}^N \tr \big(P_*^n \Sigma_{t+1}^n \big )
\Big| H^0_t
\Big]}_{\mathbb{T}_7}.\label{eq:T6T7_N}
\end{align}
\end{small}
Note that if  $\Gamma_{t+1}^n = 1$, $\hat X_{t+1}^n = X_{t+1}^n$ and  $\Sigma_{t+1}^n = 0$ and if $\Gamma_{t+1}^n = 0$, $\hat X_{t+1}^n = \hat X_{t+1|t}^n$ and $\Sigma_{t+1}^n = \Sigma_{t+1|t}^n$\footnote{If $\Gamma^n_{t+1}=0$, the remote controller gets no new information about $X^n_{t+1}$. Hence, its belief on $X^n_{t+1}$ given $H_{t+1}^0$ remains the same as its belief on $X^n_{t+1}$ given $H_{t}^0$.}. Let $\Delta$ be a random block diagonal matrix defined as follows:
\begin{align}
\Delta := \begin{bmatrix}
   \Gamma_{t+1}^1 \mathbf{I}_{d_X^1} & & \text{\huge0}\\
          & \ddots & \\
     \text{\huge0} & & \Gamma_{t+1}^N \mathbf{I}_{d_X^N}
\end{bmatrix}.
\label{Delta}
\end{align}
Then, we can write 
\begin{align}
\label{estimation_state}
&\hat X_{t+1} = \Delta X_{t+1} +(\mathbf{I}-\Delta) \hat X_{t+1|t}, \\
 &\Sigma_{t+1}^n  = (1-\Gamma_{t+1}^n) \Sigma_{t+1|t}^n, \quad n \in \mathcal{N}.
\label{estimation_covariance}
\end{align}
Now, we use \eqref{estimation_state} and \eqref{estimation_covariance} to calculate  the terms $\mathbb{T}_6$ and $\mathbb{T}_7$ in \eqref{eq:T6T7_N}. It can be shown through some straightforward  manipulations that 
\vspace{-2mm}
\begin{small}
\begin{align}
&\mathbb{T}_6= 
  \hat X_{t+1|t} ^\tp P_*^0 \hat X_{t+1|t} + \sum_{n=1}^N (1-p^n) \tr \big ([P_*^0]_{n,n}  \Sigma_{t+1|t}^n  \big).
\label{T_4_N}
\end{align}
\end{small}
Similarly, it can be shown that
\vspace{-2mm}
\begin{small}
\begin{align}\label{eq:T7_N}
&\mathbb{T}_7= \sum_{n=1}^N p^n \tr \big (P_*^n \Sigma_{t+1|t}^n  \big).
\end{align}
\end{small}
Combining \eqref{eq:T6T7_N}, \eqref{T_4_N} and \eqref{eq:T7_N}, we get
 \begin{align}
 &\ee^{g^*}[V_{t+1}|H^0_t]= \hat X_{t+1|t} ^\tp P_*^0 \hat X_{t+1|t}+ \notag \\ &\sum_{n=1}^N (1-p^n) \tr \big ([P_*^0]_{n,n}  \Sigma_{t+1|t}^n  \big) +
 \sum_{n=1}^N p^n \tr \big (P_*^n \Sigma_{t+1|t}^n  \big). \label{eq:vt1_N}
\end{align}
Since the right hand side of  \eqref{eq:vt1_N} involves  $\hat X_{t+1|t}$ and $\Sigma^n_{t+1|t}$, we will now try to write these in terms of  $\hat X_{t}$ and $\Sigma^n_{t}$. It can be easily established that 
\begin{align}
\label{estimation_X_NC}
&\hat X_{t+1|t} = A_s(0) \hat X_t, \\
&\Sigma_{t+1|t}^n =  \mathbf I + A_s(n) \Sigma_t^n A_s(n)^\tp,
\label{Sigma_t_NC}
\end{align}
where $A_s(0) = A+ B K_*^0$ and $A_s(n) = A^{nn} + B^{nn}K^n_*$ for $n=1,\ldots,N$.

Using \eqref{estimation_X_NC},\eqref{Sigma_t_NC} and the definition of $\Lambda_*$  in \eqref{eq:vt1_N}, we get
\begin{align}
 &\ee^{g^*}[V_{t+1}|H^0_t]= \hat X_t ^\tp A_s(0)^\tp P_*^0 A_s(0) \hat X_t+ \tr(\Lambda_*)\notag \\ 
   &+\sum_{n=1}^N \big((1-p^n) \tr([P_{*}^0]_{n,n}A_s(n) \Sigma_t A_s(n)^\tp) \notag \\
&+ p^n\tr(P_*^nA_s(n) \Sigma_t A_s(n)^\tp)\big). \label{eq:vt1A_N}
\end{align}

 \subsection{Simplifying the left hand side of \eqref{eq:main_goal_NC_2}}
 Adding \eqref{eq:vt1A_N} and \eqref{eq:calculated_c_N} and grouping together the terms involving $\hat X_t$ and those involving $\Sigma^n_t$, we can write
\begin{small}
 \begin{align}
& \ee^{g^*}[c(X_t,U_t^*)|H^0_t] + \ee^{g^*}[V_{t+1}|H^0_t] 
 \notag \\
& = \tr(\Lambda_*)+ \hat X_t ^\tp \Phi(P_*^0,K^0_*) \hat X_t  + \notag \\
&\sum_{n=1}^N \tr(\Phi^n((1-p^n)[P_*^0]_{n,n} +p^nP^n_*, K^n_*)\Sigma^n_t), \label{eq:vt1C_N}
\end{align}
\end{small}
where \begin{small}
\begin{align}
\Phi(P_*^0,K^0_*) := Q + (K^0_*)^\tp RK^0_*
+ (A+BK^0_*)^\tp P^0_*(A+BK^0_*),
\end{align}
\end{small}
and
\begin{small}
\begin{align}
&\Phi^n((1-p^n)[P_*^0]_{n,n} +p^nP^n_*, K^n_*) :=Q^{nn} + (K_*^n)^{\tp} R^{nn} K_*^n \notag \\
&+(A^{nn}+B^{nn} K^n_*)^\tp\big((1-p^n)[P_*^0]_{n,n} +p^nP^n_* \big)(A^{nn}+B^{nn}K^n_*).
\end{align}
\end{small}
Using the fact that $K^0_* = \Psi(P_{*}^0,R,A,B)$ and that $P^0_* = \Omega(P_{*}^0,Q,R,A,B)$, it can be shown that 
\begin{equation}
\Phi(P_*^0,K^0_*) = P^0_*.\label{eq:phi0_N}
\end{equation}
Similarly, using the fact that $K_*^n = \Psi((1-p^n)[P_{*}^0]_{n,n}+p^n P_{*}^n,R^{nn},A^{nn},B^{nn})$ and that $P_*^{n} =  \Omega \big((1-p^n)[P_{*}^0]_{n,n}+p^n P_{*}^{n},Q^{nn},R^{nn},A^{nn},B^{nn}\big)$, it can be shown that 
\begin{equation}
\Phi^n((1-p^n)[P_*^0]_{n,n} +p^nP^n_*, K^n_*) =  P^n_*.\label{eq:phi1_N}
\end{equation}

Using \eqref{eq:phi0_N} and \eqref{eq:phi1_N} in \eqref{eq:vt1C_N}, we get
\begin{small}
\begin{align}
& \ee^{g^*}[c(X_t,U_t^*)|H^0_t] + \ee^{g^*}[V_{t+1}|H^0_t]  \notag \\
&= \tr(\Lambda_*)
 + \hat X_t ^\tp P^0_* \hat X_t  + \sum_{n=1}^N\tr(P^n_*\Sigma^n_t), \notag \\
&= \tr(\Lambda_*) + V_t. \label{eq:vt1D_N}
\end{align}
\end{small}
This establishes \eqref{eq:main_goal_NC_2} and hence completes the proof of the claim.

\end{document}